\documentclass[11pt]{article}
\usepackage[margin=1in]{geometry}
\usepackage{times}

\usepackage{float}
\usepackage{amsfonts}
\usepackage{amsmath,amsthm,amssymb}
\usepackage{latexsym}
\usepackage{epic}
\usepackage{epsfig}
\usepackage{hyperref}
\usepackage{multirow}
\usepackage{hhline}
\usepackage{verbatim}
\usepackage[justification=centering]{caption}
\usepackage{enumitem}
\usepackage{color}
\allowdisplaybreaks[1]
 \usepackage[numbers]{natbib}
\usepackage{algorithm}
\usepackage{algorithmic}

\def\rev{\qopname\relax n{\mathbf{Rev}}}
\def\wel{\qopname\relax n{\mathbf{Wel}}}

 \newtheorem{theorem}{Theorem}[section]

\newtheorem{lemma}[theorem]{Lemma}
\newtheorem{proposition}[theorem]{Proposition}
\newtheorem{claim}[theorem]{Claim}

\newtheorem{remark}[theorem]{Remark}

\newtheorem{example}[theorem]{Example}

\newtheorem*{conjecture*}{Conjecture}
\newtheoremstyle{nonindented}{1ex}{1ex}{}{}{\bfseries}{.}{.5em}{}
\newtheoremstyle{indented}{1ex}{1ex}{\itshape\addtolength{\leftskip}{0.6cm}\addtolength{\rightskip}{0.6cm}}{}{\bfseries}{.}{.5em}{}
\theoremstyle{nonindented}
\theoremstyle{indented}
\theoremstyle{plain}

\newcommand{\bv}{\mathbf}

\newcommand{\set}[1]{\left\{ #1 \right\}}

\newcommand{\floor}[1]{\lfloor {#1} \rfloor}

\renewcommand{\tilde}{\widetilde}
\renewcommand{\bar}{\overline}

\DeclareMathOperator{\poly}{poly}

\def\min{\qopname\relax n{min}}
\def\max{\qopname\relax n{max}}
\def\maxs{\qopname\relax n{max2}}

\def\Pr{\qopname\relax n{\mathbf{Pr}}}
\def\Ex{\qopname\relax n{\mathbf{E}}}

\newcommand{\RR}{\mathbb{R}}

\def\R{\mathcal{R}}
\def\S{\mathcal{S}}
\def\O{\mathcal{O}}

\def\V{\mathcal{V}}
\def\X{\mathcal{X}}

\def\eps{\epsilon}

\newcommand{\INPUT}{\item[\textbf{Input:}]}
\newcommand{\OUTPUT}{\item[\textbf{Output:}]}
\newcommand{\PARAMETER}{\item[\textbf{Parameter:}]}

\newcommand{\maxi}[1]{\mbox{maximize} & {#1 } & \\}

\newcommand{\st}{\mbox{subject to} }
\newcommand{\con}[1]{&#1 & \\}
\newcommand{\qcon}[2]{&#1, & \mbox{for } #2.  \\}
\newenvironment{lp}{\begin{equation}  \begin{array}{lll}}{\end{array}\end{equation}}
\newenvironment{lp*}{\begin{equation*}  \begin{array}{lll}}{\end{array}\end{equation*}}

\begin{document}
 	
 \title{Targeting and Signaling in Ad Auctions}
 
 \author{
 	Ashwinkumar Badanidiyuru \\
 	Google\\
 	{\tt ashwinkumarbv@google.com}
 	\and
 	Kshipra Bhawalkar  \\
 	Google\\
 	{\tt kshipra@google.com}
 	\and
 	Haifeng Xu\thanks{Part of this work was conducted while Haifeng Xu was an intern at Google Research, Mountain View, CA. We would like to thank Inbal Talgam-Cohen and Konstantin Zabarnyi for pointing out an inaccuracy in the proof of Lemma \ref{lem:BayesAlgo:full} in an earlier version of the paper.} \\
 	University of Southern California\\
 	{\tt haifengx@usc.edu}	
 }

 \date{}

\begin{titlepage}
	\clearpage\maketitle
	\thispagestyle{empty}

\begin{abstract}
	Modern ad auctions allow advertisers to target more specific segments of the user population. Unfortunately, this is not always in the best interest of the ad platform -- partially hiding some information could be more beneficial for the platform's revenue. In this paper, we examine the following basic question in the context of second-price ad auctions: how should an ad platform optimally reveal information about the ad opportunity to the advertisers in order to maximize revenue?  We consider a model in which bidders' valuations depend on a random state of the ad opportunity. Different from previous work, we focus on a more practical, and challenging, situation where the space of possible realizations of ad opportunities is \emph{extremely large}. 
	We thus focus on developing algorithms whose  running time is polynomial in the number of bidders, but is  \emph{independent} of the number of ad opportunity realizations. 
		
We assume that the auctioneer can commit to a \emph{signaling scheme} to  reveal noisy information about the realized state of the ad opportunity, and examine the auctioneer's algorithmic question of designing the optimal signaling scheme. We first consider that the auctioneer is restricted to send a \emph{public} signal to all bidders. As a warm-up,  we start with a basic (though less realistic) setting in which the auctioneer knows the bidders' valuations, and show that an $\eps$-optimal scheme can be implemented in time polynomial in the number of bidders and $1/\eps$. We then move to a well-motivated Bayesian valuation setting in which the auctioneer and bidders both have private information, and present two results. First, we exhibit a characterization result regarding approximately optimal schemes and  prove that any constant-approximate public signaling scheme must use exponentially many signals. Second, we present a ``simple" public signaling scheme that serves as a constant approximation under mild assumptions. 
	
Finally, we initiate an exploration on the power of being able to send different signals \emph{privately} to different bidders.  In the basic setting where the auctioneer knows bidders' valuations, we exhibit a polynomial-time private scheme that extracts \emph{almost full surplus} even in the worst Bayes Nash equilibrium. This illustrates the surprising power of private signaling schemes in extracting revenue.

\end{abstract}

\end{titlepage}

\newpage 
\section{Introduction}

%
Every single day, through various online advertising platforms,  advertisers' ads are seen by billions of web users. Each ad opportunity, usually associated with a randomly arrived web user, can be depicted by various characteristics of the web user. A particular advantage of online advertising is that it can allow advertisers to target specific user populations. To do so, the platform (i.e., the \emph{auctioneer}) can selectively reveal some characteristics of  the web user to advertisers (i.e., \emph{bidders}). For example, the function of allowing advertisers to set different keywords and bids for different ad campaigns, as permitted by most ad platforms,  is essentially equivalent to revealing partial information to advertisers before their bidding. Features like these have motivated the following basic question in ad auctions: how much information regarding the web user should the platform reveal to advertisers? This is particularly important given increasingly refined information about web users known by today's ad platforms like Facebook or Google. In this paper, we examine this question under the objective of maximizing the auctioneer's \emph{revenue}.

Like \cite{Emek12,Miltersen12}, we focus on the single-item second-price auction and assume that each advertiser's  valuation for the ad opportunity depends on a common random \emph{state of nature}. A natural interpretation of the state of nature is the characterization of the randomly arrived web user, which affects advertisers' valuations for  the ad opportunity. The state of nature is drawn from a common-knowledge prior distribution, however only the auctioneer can observe its realization. We model the auctioneer's information revelation policy as a \emph{signaling scheme}, which is a randomized map from the realized state of nature to a set of signals. There are two natural signaling models: \emph{public} and \emph{private} signaling schemes. In a public signaling scheme, the auctioneer is constrained to send a public signal to bidders, therefore each bidder performs a Bayesian update and reaches a common posterior belief regarding the state of nature.  When the auctioneer is allowed to send different signals to different bidders, the signaling scheme is  private. 

We consider both the \emph{known valuation} and \emph{Bayesian valuation} settings. In the known valuation setting, the auctioneer knows bidders' values at any given realized state. 
Though admittedly less realistic, the known valuation setting is a natural first step for studying signaling in auctions and has been examined by several previous works \cite{Emek12,Miltersen12,Hummel15}.
In the (more realistic) Bayesian valuation setting, we assume that each bidder's value depends on a random private bidder type that is unknown to the auctioneer. To account for bidders' strategic behaviors in the auction, we adopt the dominant-strategy truth-telling equilibrium  whenever it exists, and otherwise,  the worst Bayes Nash equilibrium (particularly, in private signaling schemes).  We remark that the signaling schemes considered in  \cite{Emek12,Miltersen12} are all public. 

When designing signaling schemes for modern ad auctions, a primary computational challenge is the large diversity of web users that an ad opportunity might correspond to. Web users vary in their interests, demographics, browsing history, etc. In the advertising system, such user information is usually captured by thousands of features and each of their (exponentially many) configurations could possibly be realized. Translated to the model, this means that the space of possible realizations of the state of nature can be \emph{extremely large}.  
One key difference between our work and previous work on signaling in auctions is that we take into account this challenge of extremely large number of states of nature, and restrict our attention to algorithms with running time \emph{independent} of this parameter.  


We note that signaling naturally relates to  \emph{targeting} in ad auctions. In fact, as we will elaborate later, the Bayesian valuation setting we consider is motivated by the ad auction feature called \emph{re-targeting}.  A signaling scheme essentially specifies the extent to which each advertiser can target the web users.   Unsurprisingly, it is usually not  in the best interest of the auctioneer to reveal full information to advertisers: over refined targeting can result in fewer advertisers competing for the ad opportunity and possible a large gap between the highest and second-highest bid  \cite{Hummel15,Korula2016}. This is detrimental to the revenue in second price auctions. On the other hand, completely obscuring user information is also suboptimal since advertisers will bid conservatively due to their uncertainty regarding the value of the ad opportunity. Our goal is thus to design optimal  signaling schemes that  properly balance the amount of competition and bidders' valuations in the auction.   Next, we elaborate our results in detail.\footnote{Note that though our work is primarily motivated by ad auctions, our results are applicable to any second-price auction with large number of possible item realizations.} 

\subsection{Our Models and Results}
We first consider the design of optimal public signaling schemes. As a warm-up, we start with the known-valuation setting and show that a simple Monte-Carlo algorithm implements an $\eps$-optimal public scheme and runs in $\poly(n, \frac{1}{\eps})$ time where $n$ is the number of bidders. Note that  this setting is also studied by  \cite{Emek12,Miltersen12}, however the running time of their algorithms depends polynomially  on the number of states of nature.  

Next, we move to a natural Bayesian-valuation model motivated by \emph{re-targeting}. Re-targeting is a feature in ad auctions that allows each advertiser to target a set of web users (e.g., those who visited the advertiser's website) and use a different bid for users from that set.\footnote{Re-targeting is allowed in several ad platforms including Google's and Facebook's. Each advertiser will specify a targeted set of web users (usually using their user IDs) in the platform system. When a web user arrives, the platform informs each bidder whether the user is in their targeted set or not. Then each bidder submits a bid. We refer curious readers to the help pages from various companies for more details, e.g., \url{https://support.google.com/adwords/answer/1704368?hl=en}.} In this setting, the state of nature is a \emph{vector} of binary feature $\theta_i \in \{0,1\}$ for each bidder $i$,  indicating whether the web user is the bidder's targeted ones or not. Therefore, the number of possible states is exponential in the number of bidders. We assume that  bidder $i$'s valuation is determined by $\theta_i$ and a private bidder type that is i.i.d. across bidders. Note however, $\theta_i$'s may be \emph{correlated} since some bidders may target similar items.  For this Bayesian-valuation model, we first prove that any public signaling scheme that uses sub-exponentially many signals can \emph{not} guarantee any constant fraction of the optimal revenue. 
We note that, though \emph{not} a measure of computational complexity, the number of signals can be viewed as a natural complexity measure of the signaling scheme, which is analogous to the menu-size complexity studied in the optimal mechanism design literature \cite{Hart2013,Wang2014,Babaioff2017}. Complementing this result, we  show a  signaling scheme that serves as a constant approximation under mild assumptions. Our scheme is ``simple" to implement and tracks the following conventional wisdom. In particular, it reveals full information to bidders whenever there is  enough competition in the auction and otherwise, it pools bidders to form competition by selectively obscuring information.  

Finally, we initiate an exploration on the power of private signaling schemes, which to the best of our knowledge has not been studied in the context of auctions.  
We first exhibit an example showing that the optimal private scheme can outperform the optimal public scheme \emph{arbitrarily} in terms of the auctioneer's revenue.  This motivates the algorithmic question of computing the optimal private signaling scheme. 
Here, we restrict our attention to the known-valuation setting, which is basic but  already reveals sufficient non-triviality due to intricate equilibrium bidding behaviors and the equilibrium selection issue.  
We present  a polynomial-time private signaling scheme that extracts almost the \emph{full surplus} even in the \emph{worst} Bayes Nash equilibrium. This illustrates the surprising power of private signaling schemes in extracting revenue.


\subsection{Additional Related Work}
The study of information asymmetry in auctions dates back as early as 1980s \cite{Rothkopf1969,Milgrom1979,Milgrom1982b,Milgrom1982}.  In a seminal paper, Milgrom and Weber \cite{Milgrom1982} prove that in common-value auctions with bidders of symmetric random signals, among all public signaling schemes, revealing full information maximizes the auctioneer's revenue in most commonly seen auction formats (i.e., descending, ascending, second-price and first-price auctions). However, later work has observed that such full transparency is not optimal in general, and the auctioneer could increase revenue by revealing partial information to bidders \cite{Perry1999,Campbell2000,Feinberg2005,Syrgkanis2013}. This has motivated a surge of recent interest in the algorithmic question of computing the optimal signaling scheme that maximizes revenue in a given auction format.  \citet{fu2012} prove that revealing full information is optimal, among all public schemes, in Myerson's optimal auction. However, the optimal scheme is usually a non-trivial one and involves randomization in second price auctions \cite{Emek12,Miltersen12}. Recently, \citet{Daskalakis2016}  characterize the optimal joint design of the signaling scheme and auction mechanism. 

To our knowledge, \citet{Emek12,Miltersen12} are the first to study optimal signaling design in the known-valuation setting. They observe  a linear program for computing the optimal public scheme in time polynomial in the number of bidders and states of nature. However, this is not satisfactory for ad auctions since the space of all possible states of nature is usually extremely large. It is generally impractical to solve such a large LP quickly.  Though not explicitly stated, \citet{Hummel15} also study the known-valuation setting with exponentially many states of nature. However, they do not consider the design of optimal signaling but instead compare the revenue between two simple schemes, i.e., revealing no information or full information. 

\citet{Emek12} also study the Bayesian valuation setting, though in a more idealized model. Specifically, they assume that given a state of nature, the auctioneer's uncertainty regarding the bidders's valuation is captured by a distribution over polynomially many value profiles. This is somewhat restrictive since the space of value profiles typically increases exponentially in the number of bidders.  However, even under this simplification, it is shown that computing the optimal public signaling scheme is NP-hard. In a follow-up work, \citet{Cheng2015} derive a PTAS for the Bayesian model of \cite{Emek12}. Their algorithm crucially relies on the fact that there are polynomially many states of nature, thus does not apply to our setting.  

All the mentioned algorithmic work so far focuses on public signaling. Recent work of \cite{Arieli2016,Babichenko2016b,Dughmi2017} studies  optimal private signaling for agents with no inter-dependence. Therefore, given a signal, each agent makes a decision that is independent of others. The problem usually becomes more intricate when there are strategic interactions among agents (e.g., in an auction). \citet{Bergemann2016} characterize the distributions of players' action profiles at equilibrium under an asymmetric information structure  as the set of Bayes correlated equilibria. To our knowledge, there has not been much progress on  the design of  optimal private signaling in such  strategic settings. \citet{Taneva2015} characterizes the optimal private signaling in a two-agent two-action game.  \citet{Vasserman2015} consider private signaling in non-atomic routing, and focus on analyzing the cost decrease via private signaling in a graph of parallel links.

\section{Preliminaries}
\label{sec:prelim}
\subsection{The Model}
We focus on the single-item second price auction, which is widely adopted by ad exchange platforms for selling display ads, video ads, etc. Let $n$ denote the total number of bidders (``he" or ``them"). Vector $v = v(\theta) \in \RR_+^n$ denotes the profile of bidders' valuations which depends on a random \emph{state of nature}  $\theta$. 
We assume that the state of nature is a-priori unknown to the bidders, and drawn from a common-knowledge prior distribution $\lambda$ supported on set $\Theta$. The auctioneer (``she"), on the other hand, can observe the realized $\theta$.  For ease of presentation, we assume $\Theta$ is finite, though this is not crucial for our results.  In ad auctions, $\theta$ could encode the information of a web page viewer (e.g., demographic information, browsing history, etc.), which  is known to the auctioneer.  We call this the {\bf Known-Valuation Setting} (KVS) since in this case, the auctioneer is certain about each bidder's value after observing $\theta$. We use $v_i(\theta)$ to denote bidder $i$'s value and $\V_i = \{ v_i(\theta): \theta \in \Theta \}$ is its support. Naturally,  $ v(\theta) \in \V = \prod_{i=1}^n \V_i$ for any $\theta$. 
We assume the set $\Theta$ to be \emph{extremely large}, and will focus on algorithms with running time independent of $|\Theta|$. 
However, any algorithm can query the distribution $\lambda$ for the probability $\lambda_{\theta}$ of a state $\theta$ and can also sample from $\lambda$. 

We also study the {\bf Bayesian-Valuation Setting} (BVS), in which bidders' values have two-sided uncertainties. Motivated by the \emph{re-targeting} feature in ad auctions, we focus on a basic model in which the state of nature consists of a \emph{binary} feature $\theta_i \in \{0,1 \}$ for each bidder $i$, indicating whether the item is among bidder $i$'s targeted ones or not. Therefore, $ \theta  = (\theta_1,...,\theta_n) \in \{0,1 \}^n = \Theta$. We assume bidder $i$'s value $v_i = v(\theta_i, t_i)$ is a function of $\theta_i$ and  a private bidder type  $t_i $, which bidder $i$ is privy to, but is unknown to the auctioneer. Instead, the auctioneer only knows the prior distribution of $t_i$. 
We assume that the bidder type $t_i$'s  are independent and identically distributed (i.i.d.).  Therefore, $v_i = v(1,t_i)$ follows the same distribution for each $i$, denoted as $H(v)$, and  $v_i = v(0,t_i)$ follows another distribution $L(v)$. We also call $H [L]$ the \emph{high} [\emph{low}] distribution for convenience, since it corresponds to the value distribution of targeted [un-targeted] items. W.l.o.g., assume $\Ex_{v \sim H} [v] > \Ex_{v \sim L} [v] $. Notice however, we allow the binary feature $\theta_i$'s  to be correlated  since some bidders may target similar items. If, moreover, $\{ \theta_i \}_{i \in [n]}$ are also i.i.d., we say the bidders  are \emph{symmetric}.

\subsection{Signaling Schemes}
We assume that the auctioneer can commit to a \emph{signaling scheme} that  reveals noisy information about the realized $\theta$. A signaling scheme $\varphi$ is a \emph{randomized} map  from  the states in $\Theta$ to a set of \emph{signals}  $ \Sigma$. We use $\varphi(\theta,\sigma)$ to denote the probability of selecting the signal $\sigma \in \Sigma$ given the state $\theta$. Notice that $\varphi$ is public knowledge. An algorithm implements signaling scheme $\varphi$ if it takes as input a state of nature $\theta$, and outputs $\sigma$ with probability $\varphi(\theta,\sigma)$.  Our goal is to design efficient and (approximately) optimal \emph{revenue-maximizing} signaling schemes. When $\varphi$ yields expected revenue within an additive [multiplicative] $\epsilon$ of the best possible, we say it is $\epsilon$-optimal [$\epsilon$-approximate] in the additive [multiplicative] sense. We consider  two natural  signaling models. 


{\bf Public Signaling Schemes}. In a \emph{public} signaling scheme, the auctioneer is only allowed to send a public signal, so that each bidder infers the same information. Let $\alpha_{\sigma} = \sum_{\theta \in \Theta} \lambda_{\theta} \varphi(\theta,\sigma)$ denote the probability of sending the public signal $\sigma$.  Upon receiving $\sigma$, bidders infer a common posterior distribution of $\theta$ by Bayesian updating. In particular,  $\Pr(\theta|\sigma) = \lambda_{\theta} \varphi(\theta,\sigma)/\alpha_{\sigma}$. In the known-valuation setting, we use $$v(\sigma) = \Ex_{\theta | \sigma} [v(\theta)]= \frac{1}{\alpha_\sigma} \sum_{\theta} \lambda_{\theta} \varphi(\theta,\sigma) \cdot  v(\theta)$$ to denote the expected value profile conditioned on signal $\sigma$, and  $ v_i(\sigma)$ is bidder $i$'s conditional expected value. Observe that, conditioned on receiving the public signal  $\sigma$, bidding $v_i(\sigma)$ is the dominant strategy for each bidder $i$ (see, e.g., \cite{Emek12}) in the second price auction. 
Our goal is to maximize the auctioneer's expected revenue at the (unique) dominant-strategy truth-telling equilibrium, expressed as follows in the known-valuation setting: 
\begin{equation*}
	\mbox{KVS: } \quad \rev(\varphi) = \sum_{\sigma \in \Sigma} \alpha_{\sigma} \cdot \maxs_{i} \,  [v_i(\sigma)] = \sum_{\sigma \in \Sigma} \maxs_{i}  \, \big[ \sum_{\theta} \lambda_{\theta} \varphi(\theta, \sigma) \cdot v_i(\theta) \big] \qquad 
\end{equation*}
where the function $\maxs$ returns the second largest value of a given set or vector. The following lemma shows that an optimal public signaling does not need to use more than $n(n-1)$ signals in the known-valuation setting. The proof follows a simple revelation-principal-type argument.   
\begin{lemma}\label{lem:revelation}(Adapted from \cite{Kamenica2011,Emek12})
	There always exists an optimal \emph{public} signaling scheme in the \emph{known-valuation} setting  that uses at most $n(n-1)$ signals, with signal $\sigma^{ij}$ ($i \not = j \in [n]$) resulting in bidder $i$ $[j]$ as the bidder with the highest [second-highest] value. 
\end{lemma}

In the Bayesian-valuation setting, the description of a public signaling scheme remains the same. Therefore, the expected revenue is  expressed similarly, but with an additional expectation over the randomness of bidder's private types.
\begin{equation*}
	\mbox{BVS: } \qquad \rev(\varphi) = \Ex_{t_1,...,t_n} \bigg[  \sum_{\sigma \in \Sigma} \maxs_{i}  \, \big[ \sum_{\theta} \lambda_{\theta} \varphi(\theta, \sigma) \cdot v(\theta_i, t_i) \big] \bigg]\qquad 
\end{equation*}

We remark that Lemma \ref{lem:revelation} does not hold in the Bayesian-valuation setting. Actually, in sharp contrast to Lemma \ref{lem:revelation},  we prove in Section \ref{sec:Bayes} that in BVS,  any public signaling scheme using sub-exponentially many signals cannot guarantee even a constant fraction of the optimal revenue.  

Note that  among all public signaling schemes, the one that maximizes the \emph{social welfare} is to reveal full information to bidders, in both known-valuation and Bayesian-valuation settings \cite{Dughmi2014}. Moreover, the optimal welfare is always an upper bound for the optimal revenue of public signaling. 
We remark that  though the signaling model assumes that each bidder knows the prior distribution and signaling scheme, in practice bidder $i$ only needs to know his conditional valuation $v_i(\sigma) = \Ex_{\theta|\sigma} v_i(\theta)$ on signal $\sigma$ in order to figure out his optimal bid.  $v_i(\sigma)$ can be estimated within good precision by  samples (i.e., interacting with the system for enough times). 

{\bf Private Signaling Schemes}. In a private signaling scheme, the auctioneer is allowed to send different (possibly correlated) signals to different bidders. That is, bidders are asymmetrically informed. After receiving  the private signal, each bidder updates his own belief regarding $\theta$. Obviously, the public signaling scheme is a special class of private signaling schemes.  We will restrict our attention to the \emph{known-valuation setting} when  studying private signaling schemes. An important observation regarding private signaling is that bidding the conditional expected value is generally not an equilibrium strategy for bidders. Moreover, the phenomenon of winner's curse may arise when bidders are asymmetrically informed, as illustrated in the following example.   

\begin{example}[Non-Truthfulness and Winner's Curse in Private Signaling]\label{ex:curse}
	Consider a second-price auction with $2$ bidders in the known-valuation setting. There are two possible value profiles, as listed in the following table together with corresponding probabilities.
	\begin{center}
		\begin{tabular}{|c|cc|}
			\hline
			& bidder 1 & bidder 2  \\
			\hline
			\emph{prob} $0.5$ & $1$ & $2$ \\
			\hline
			\emph{prob} $0.5$  & $7$  & $7$ \\
			\hline
		\end{tabular}
	\end{center}
	
	
	Consider the private signaling scheme that reveals the realized profile (i.e., full information) to bidder $2$ but reveals no information to bidder $1$. Since bidder $2$ knows her value exactly, bidding the true value is the dominant strategy for bidder $2$. Given no information, bidder $1$ has an expected value of $4$, which however is not  her equilibrium bid. In fact, bidder $1$ will not bid any value between $2$ and $7$ since in this case she can only win at the less-valued situation -- the winner's curse.  The equilibrium bid for bidder $1$ turns out to be any bid under $2$ and there are infinitely many equilibria under this private signaling scheme. 
\end{example}

Example \ref{ex:curse} reveals several challenges for computing the optimal private signaling scheme. The first is the equilibrium selection problem due to the existence of multiple equilibria,  as illustrated by Example \ref{ex:curse}. There is usually not a natural equilibrium selection rule in private signaling.  Second, even if we know which equilibrium to adopt, computing a desired equilibrium is non-trivial in second-price auctions with  asymmetrically informed bidders. In fact, most prior work is only able to characterize equilibria under restrictive assumptions, e.g., two bidders, common valuation \cite{Hausch1987,Fang2002,Syrgkanis2013}.  

\section{Warm-up: Public Signaling in the Known-Valuation Setting}
\label{sec:KnownPublic}
We start with the design of the optimal public signaling scheme in the  known-valuation setting, and prove the following theorem.

\begin{theorem}\label{thm:public}
	Assume $v(\theta) \in [0,1]^n$ for all $\theta$. An (additive) $\epsilon$-optimal signaling scheme can be implemented in $\poly(n,\frac{1}{\epsilon})$ time. In particular, the running time is independent of $|\Theta|$. 
\end{theorem}

The proof of  Theorem \ref{thm:public} combines and generalizes two ideas from previous literature. We present a sketch here, while defer the full proof to Appendix \ref{appendix:KnownPublic}. First, as observed in \cite{Emek12,Miltersen12},  there is a $\poly(n,|\Theta|)$-sized linear program that explicitly computes the optimal public signaling scheme. Specifically,  by Lemma \ref{lem:revelation},  we can w.l.o.g. assume $\Sigma = \{ \sigma^{ij}\}_{i\not = j \in [n]}$ in which $\sigma^{ij}$ results in bidder $i$ [$j$] to  have the highest [second-highest] expected value. Then Linear Program  \eqref{lp:public} with variables $\varphi (\theta,\sigma)$ computes the optimal public signaling scheme. However, any naive approach for solving LP \eqref{lp:public} requires  $\poly(n,|\Theta|)$ time, which is prohibitive when $|\Theta|$ is extremely large.  To overcome this challenge, we employ a similar Monte-Carlo sampling algorithm as proposed by \citet{Dughmi2016}. In particular, given as input a state of nature $\theta$, we take $\poly(n,\frac{1}{\eps})$ many samples from $\lambda$, and let $\tilde{\lambda}$ be the empirical distribution over these samples plus the given state $\theta$.  The algorithm then solves the linear program by instead  using  $\tilde{\lambda}$ as the prior and  relaxing the first two constraints of LP \eqref{lp:public} by $1/\poly(1/ \epsilon, n)$ (since the samples cannot exactly preserve these constraints), and signals for $\theta$ as suggested by the solution of the LP.

\begin{lp}\label{lp:public}
	\maxi{\sum_{\sigma^{ij} \in \Sigma}  \sum_{ \theta \in \Theta} \lambda_{\theta} \varphi(\theta,\sigma^{ij}) v_j(\theta) }
	\st 
	\qcon{\sum_{\theta \in \Theta} \lambda_{\theta} \varphi(\theta,\sigma^{ij}) v_i(\theta) \geq \sum_{\theta \in \Theta} \lambda_{\theta} \varphi(\theta,\sigma^{ij}) v_j(\theta) }{j \not = i}
	\qcon{\sum_{\theta \in \Theta} \lambda_{\theta} \varphi(\theta,\sigma^{ij}) v_j(\theta) \geq \sum_{\theta \in \Theta} \lambda_{\theta} \varphi(\theta,\sigma^{ij}) v_k(\theta) }{k \not = i,j ; \, j \not = i}
	\qcon{\sum_{\sigma^{ij} \in \Sigma} \varphi (\theta,\sigma) = 1}{\theta \in \Theta}
	\qcon{\varphi (\theta,\sigma^{ij}) \geq 0}{\theta \in \Theta, \sigma^{ij} \in \Sigma}
\end{lp} 	

We remark that though the idea for our Monte-Carlo algorithm is similar to that in  \cite{Dughmi2016}, the analysis there does not directly apply to our setting. Particularly, the algorithm  in \cite{Dughmi2016} gives a \emph{bi-criteria} FPTAS for the Bayesian persuasion problem (i.e., $\eps$-loss in both optimality and incentive compatibility), and they prove that the bi-criteria loss is inevitable in their setting due to information-theoretic reasons. However, our algorithm is a \emph{single-criteria} FPTAS  and its analysis relies on the special property of the revenue objective.

\section{Public Signaling in the Bayesian-Valuation Setting}\label{sec:Bayes}
In this section, we turn our attention to the  more realistic Bayesian valuation setting, and consider the design of optimal public signaling schemes.  We focus on a basic yet well-motivated model, i.e., $ \theta  = \{ \theta_1,...,\theta_n\} \in  \{0,1 \}^n$ and bidder $i$'s private information is captured by a type $t_i$ which is i.i.d. across bidders. Note however, $\theta_i$'s can be correlated.
When $\theta_i = 1$,  bidder $i$'s value $v_i = v(1, t_i)$ follows a common (high) distribution $H(v)$ for every $i$. Similarly, $v_i = v(0,t_i)$ follows the (low) distribution $L(v)$. 
We assume w.l.o.g  $\Ex_{v \sim H} [v] > \Ex_{v \sim L} [v] $.  If, moreover, $\theta_i$'s are i.i.d., we say that bidders are \emph{symmetric}.  

\subsection{Lower Bounding the Number of Signals Required by Approximately Optimal Schemes}
We first bound the minimum number of signals required by any approximately optimal public signaling scheme. This is an important feature not only because it is a basic property regarding the optimal scheme but also it serves as a natural complexity measure of the optimal scheme, which is analogous to the menu-size complexity  studied in the recent  mechanism design literature \cite{Hart2013,Wang2014,Babaioff2017}.  Several previous works have established polynomial-size \emph{upper bounds} on the number of signals required by the optimal scheme in various settings  \cite{Kamenica2011,Emek12,Miltersen12}. 
In contrast, we prove that in our Bayesian-valuation setting,   it is generally necessary to use \emph{exponentially} many signals, even when only a constant-approximate public signaling scheme  is sought.


\begin{theorem}\label{thm:NumSignal}
	In the Bayesian-valuation setting with  $n$ symmetric bidders, any public signaling scheme using $2^{o(n^{1/4})}$ signals cannot guarantee any constant (multiplicative) fraction of the optimal revenue. 
\end{theorem}

We remark that Theorem \ref{thm:NumSignal} is \emph{not}  a computational hardness result, though it does indicate some challenges in designing the optimal public signaling scheme for the BVS setting. The proof of Theorem \ref{thm:NumSignal} is rather involved. We defer the full proof  to Appendix \ref{appendix:NumSignalsA}, and only present a sketch here. 

The proof examines the following instance in the symmetric bidder setting: there are $n$ bidders and each bidder's valuation follows either a (high) Bernoulli distribution with success probability $\frac{1}{\sqrt{n}}$ (when $\theta_i = 1$) or a (low) point distribution at $0$  (when $ \theta_i = 0$); $\theta_i$'s are i.i.d. with $\Pr(\theta_i = 1) = \eps$ for a small constant $\eps$. 

The proof is divided into two steps. We first show that the revenue of the optimal public signaling scheme is at least $1 - o(1)$. Intuitively, this is because when revealing full information, with high probability there are $\Omega (\eps n)$ bidders having the high value distribution and such competition results in a revenue close to $1$.  Second, we argue that the revenue of any public signaling scheme using  $2^{o(n^{1/4})}$ signals must be far less than $1$.   The main technical challenge we overcome is in the second step: given an \emph{arbitrary} public signaling scheme using at most $N$ signals, we look to upper bound its revenue in terms of $N$. 

In fact, instead of upper bounding the revenue, we upper bound the welfare of an arbitrary public signaling scheme using at most $N$ signals since the revenue is then trivially upper bounded by the welfare.  We first observe that among the public signaling schemes with at most $N$ signals  that achieve the maximum welfare, there is always a deterministic one that deterministically maps a state to a signal. Such a deterministic scheme can be viewed as a partition of $\Theta = 2^{[n]}$ into $N$ subsets, denoted as $\{S_i \}_{i=1}^N$. Let $\Pr(S_i) = \sum_{\theta \in S_i} \lambda_{\theta}$ denote the total probability mass of the states in set $S_i$. We then prove the following key lemma (informal): for any $S_i$, either the welfare conditioned on $S_i$ is $O(\eps)$ or $\Pr(S_i) $ is exponentially small (specifically, $ \Pr(S_i) \leq 2^{-\Theta(n^{1/4})}$).  This lemma implies that any deterministic scheme with $2^{o(n^{1/4})}$ signals must have $O(\epsilon)$ welfare.  The proof of the lemma involves intricate combinatorial analysis. At a high level, we view each state $\theta$ as a binary vector and let $\bar{\theta}$ denotes the expectation of the $\theta$'s in $S_i$. Note that each entry of $\bar{\theta}$ has expectation $\epsilon$ a-priori. We show that if $m$ is the number of entries in $\bar{\theta}$ that are at least $2\eps$, then $\Pr(S_i) \leq e^{-2 \sqrt{m}}$. Intuitively, this is because the sum of the expectation of these $m$ entries is $\eps m$, so the probability that the sum is at least $2 \eps m$ should be small. However, the concrete argument is more involved since we have to argue for an arbitrary deterministic set $S_i$ while not a sampled one. Since the high distribution is a Bernoulli with success probability $1/\sqrt{n}$, if $m = o(\sqrt{n})$, the welfare from $\bar{\theta}$ would be $O(\eps)$; otherwise $m = \Omega(\sqrt{n})$ implying $\Pr(S_i) \leq 2^{-\Theta(n^{1/4})}$. This yields a proof sketch of the key lemma. 

\subsection{A Constant Approximate Scheme} 
In this section, we present a ``simple" scheme that serves as a constant approximation under mild assumptions. We start by showing that revealing full information -- which indeed uses exponentially many signals -- can be arbitrarily far from optimality. 

\begin{example}[Failure of the Full Information Scheme]\label{ex:failFull}
	Consider the following simple example: the value of each bidder either follows the (high) uniform distribution in $[0,1]$ (denoted as $U[0,1]$) or follows the (low) point distribution at $0$. The state of nature is a uniformly drawn bidder whose value follows $U[0,1]$ while all other bidders have valuation $0$. Only the auctioneer knows which bidder has a high value distribution.  If the auctioneer reveals full information, the revenue is $0$ since only one bidder has a non-zero valuation. On the other hand, consider the following signaling scheme: send a \emph{pooling signal} $\sigma_{2i-1,2i}$ for $i \in {1,2,...,\floor{n/2}}$ when  either bidder $2i-1$ or $2i$ has the high value distribution. That is, the scheme pools bidder $2i-1$ and bidder $2i$ to create competition.  Now from bidders' perspective, whenever the pooling signal $\sigma_{2i-1,2i}$ is sent, bidder $2i-1$ and $2i$ know that  their valuations are drawn from $U([0,1])$ with probability $0.5$ and are $0$ otherwise. As a result, bidder $2i-1$ and $2i$ have expected value distribution $U([0,0.5])$.  Some calculation reveals that the expected revenue conditioned on signal $\sigma_{2i-1,2i}$ is $1/6$, which is arbitrarily better than $0$.  
\end{example}

Revisiting Example \ref{ex:failFull}, we note that the failure of the full information scheme there is due to the lack of competition. In particular, by revealing full information, the realized state includes  only one non-zero-valued bidder, resulting in $0 $ revenue. We call these the \emph{tail states}, meaning the extreme situation with little competition in the auction. More specifically, we call the states in $\Theta_1 = \{\theta: |\theta| = 1 \}$ the tail states.  The high-revenue scheme described in Example \ref{ex:failFull} pools bidders to create competition at these tail states. Obviously, this is not always possible for any prior distribution $\lambda$. For example,  if  $\lambda$ is a point distribution supported at the  tail state $\theta = (1,0,0,...,0)$, then we cannot pool $\theta$ with any other tail states. 

We thus first characterize when such tail-pooling is possible. Since the state of nature $\theta$ is in $\{0,1 \}^n$, for convenience we sometimes view $\theta$ as a \emph{set of bidders} specified by entries of value $1$. Therefore, each tail state can be treated as a bidder and $\Theta_1$  can be treated as the set of all bidders $ [n]$. We would like to pair the tail states  in  $\Theta_{1}$ (possibly randomly)   in a way such that each bidder in a pair has \emph{equal} probability mass (like the scheme in Example \ref{ex:failFull}).  Such a \emph{pooling scheme} can be captured by a randomized map $\pi : \Theta_1 \to \Theta_1$ which maps any given state $\theta \in \Theta_1$ to the state that it will be paired with. Let $\pi(\theta, \theta')$ denote the probability of mapping $\theta$ to $\theta'$. The \emph{feasible} pooling scheme $\pi$'s are captured by the following linear system.
\begin{lp}\label{lp:balance}
	\qcon{\sum_{\theta' \in \Theta_1} \pi(\theta, \theta')= 1}{i \in \Theta_{1} }
	\qcon{ \lambda_{\theta} \cdot \pi(\theta, \theta') = \lambda_{\theta'} \cdot \pi(\theta', \theta) }{ \theta,  \theta'  \in \Theta_1} 
	\qcon{\pi(\theta, \theta') \geq 0 \, \, \,  \mbox{ and }  \, \, \,  \pi(\theta, \theta) = 0}{\theta ,\theta'  \in \Theta_1 }
\end{lp}    

Particularly, the second constraint guarantees that conditioned on that $\theta, \theta'$ are pooled together, they will have equal probability mass. As pointed out earlier, linear system \eqref{lp:balance}  is not always feasible and its feasibility depends on the portion of $\lambda$ on set $\Theta_1$. The following lemma shows that under a minor assumption on $\lambda$, linear system \eqref{lp:balance}  is guaranteed to be feasible. The proof of Lemma \ref{lem:tail_balance} can be found in Appendix \ref{appendix:BayesAlgLem}. 
\begin{lemma}\label{lem:tail_balance}
	Linear system \eqref{lp:balance} is feasible if and only if for any $\theta \in \Theta_{1}$, $\lambda_{\theta} \leq \sum_{\theta' \in \Theta_{1}: \theta' \not = \theta} \lambda_{\theta'}$. 
\end{lemma}   

In other words, there exists a feasible pooling scheme if no state in $\Theta_{1}$ has probability mass that is larger than the sum of \emph{all other} states' probability masses. In this case, we say $\lambda$ is tail-balanced. Particularly, $\lambda$ is \emph{tail-balanced} if  $\lambda_{\theta} \leq \sum_{\theta' \in \Theta_{1}: \theta' \not = \theta} \lambda_{\theta'}$ for any $\theta \in \Theta_{1}$. It is easy to see that a feasible pooling scheme (if exists)  can be computed efficiently.   Utilizing any feasible pooling scheme, we consider the signaling scheme that simply reveals $\theta$ \emph{unless} $\theta \in \Theta_1$ in which case it pools $\theta$ with another $\theta' \in \Theta_1$ using a pooling scheme. The details are in Algorithm \ref{alg:publicBayes}, which we term the \emph{tail-pooling signaling scheme}.

\begin{algorithm}
	\begin{algorithmic}[1]
		\PARAMETER any feasible pooling scheme $\pi$. 
		\INPUT State of nature $\theta \in \{ 0,1\}^n$.  
		\OUTPUT Signal $\sigma$ 
		
		\vspace{1mm}
		
		\IF{$|\theta|_1 \geq 2 $ or $|\theta|_1 = 0$} 
		\STATE $\sigma = \theta$;
		\ELSE
		\STATE Pooling tail states: sample $\theta'  \in \Theta_{1}$ with probability $\pi(\theta, \theta')$, and let $\sigma = \{ \theta,\theta'\}$; 
		\ENDIF      
		\RETURN $\sigma$.
	\end{algorithmic}
	\caption{Tail-Pooling Signaling Scheme for  Bayesian Valuation Setting}
	\label{alg:publicBayes}
\end{algorithm}


The output of Algorithm \ref{alg:publicBayes} simply either reveals the true state (when $\sigma = \theta$) or reveals a set of two states (when $\sigma = \{\theta, \theta' \}$). Conditioned on  $\sigma = \{\theta, \theta' \}$, the state $\theta$ and $\theta'$ show up both with probability $1/2$, as imposed by the second constraint in linear system \eqref{lp:balance}. 
The following theorem shows that under mild assumptions, Algorithm \ref{alg:publicBayes} serves as a constant approximation to the optimal public signaling scheme.

\begin{theorem}\label{thm:unknownAlgo}
	Assume the high and low distributions satisfy the monotone hazard rate (MHR) condition and $\lambda$ is tail-balanced.  When $n \geq 22$, Algorithm \ref{alg:publicBayes} achieves revenue that is at least  $\frac{1}{8}$ fraction of the \emph{optimal welfare}, thus is a  $\frac{1}{8}$-approximate revenue-maximizing signaling scheme. 
\end{theorem}


Observe that the optimal welfare is obtained by revealing full information, thus is the expectation of the welfare at each state. Theorem \ref{thm:unknownAlgo} is proved by showing that for every state $\theta \in \Theta$, the revenue at $\theta$ in Algorithm \ref{alg:publicBayes} is at least $1/8$ fraction of the welfare at $\theta$ in the full information scheme. The following lemma summarizes the concrete bounds for different cases. Utilizing MHR properties, our proof  establishes various inequalities comparing the welfare and revenue of a standard second price auction when each bidder's value is drawn independently from the high or low value distribution. The details are relegated to Appendix \ref{appendix:BayeAlgo}. 

\begin{lemma}\label{lem:unknownAlgo}
	For any $\theta \in \Theta$, let $\rev_{pool}(\theta)$ denote the revenue of the tail-pooling signaling scheme in Algorithm \ref{alg:publicBayes} at state of nature $\theta$ and $\wel_{full}(\theta)$ denote the welfare of the full information scheme at state of nature $\theta$. Then:  
	\begin{enumerate}
		\item When $|\theta| = 0$ and $n \geq 2$, $\rev_{pool}(\theta) \geq \max \{\frac{1}{3}, \frac{\ln(n+1) - 1}{ \ln(n+1)} \} \cdot  \wel_{full}(\theta)$. 
		\item When $|\theta| \geq 2$ and $n \geq 22$, $\rev_{pool}(\theta) \geq \max \{\frac{1}{6}, \frac{\ln(|\theta|+1) - 1}{ 2\ln(|\theta|+1)} \} \cdot  \wel_{full}(\theta)$. 
		\item When $|\theta| = 1$ and $n \geq 22$, $\rev_{pool}(\theta) \geq \frac{1}{8} \cdot  \wel_{full}(\theta)$. 
	\end{enumerate}
\end{lemma}

\begin{remark}
	Theorem \ref{thm:unknownAlgo} conveys an interesting conceptual message. That is, the simple ``rule of thumb" principle, which reveals full information when there is enough competition and otherwise pools several bidders to form competition, is approximately optimal. As an instantiation of this principle, Algorithm \ref{alg:publicBayes} pools the tail states consisting of a single high-valuation bidder to form competition among \emph{two} bidders. We make this choice because such a pooling scheme exists under only a minor assumption on $\lambda$, i.e., tail balance (see Lemma \ref{lem:tail_balance}). Nevertheless, from a practical perspective, if there exist pooling schemes that can form competition among more bidders by pooling the states consisting of a few high-valuation bidders, we suspect that the analogy of Algorithm \ref{alg:publicBayes} may yield more revenue in practice.       
\end{remark}

\section{Private Signaling in the Known-Valuation Setting}
\label{sec:PublicPrivate}
In this section, we consider the design of the optimal private signaling scheme in the known-valuation setting. To see the advantage of using private  schemes, we first exhibit an example showing that private signaling schemes may achieve \emph{arbitrarily} better revenue than the optimal public signaling scheme. 

\begin{example}\label{ex:AggreBid}
	Consider a second-price auction with $3$ bidders in the known-valuation setting. There are two value profiles, as listed in the following table  with corresponding probabilities.
	\begin{table}[H]
		\centering
		\begin{tabular}{|l|ccc|}
			\hline
			& bidder 1 & bidder 2 & bidder 3 \\ \hline 
			prob: $1-\epsilon$ & $2\epsilon$  & $\epsilon$ & 1 \\ \hline
			prob:  $\epsilon$ & 1 & $1 - \epsilon$ & $\epsilon$ \\
			\hline
		\end{tabular}
	\end{table}
\end{example}


\begin{proposition}\label{prop:gap}
	Assume any bidder bids her expected value when this is the dominant strategy. In Example \ref{ex:AggreBid}, there exists a private signaling scheme that achieves arbitrarily better revenue, even in the \emph{worst} Bayes Nash equilibrium, than the optimal public signaling scheme.
\end{proposition}

Proposition \ref{prop:gap} follows from Claim \ref{claim1:gap} and  Claim \ref{claim2:gap}. An intuitive explanation is as follows.  In Example \ref{ex:AggreBid}, bidder $1$ values the item more than bidder $2$ at both states. We say bidder $1$ is \emph{stronger} than bidder $2$. We consider a private signaling scheme that simply reveals full information to bidder $2$ and $3$  but no information to bidder $1$. This results in bidder $1$, whose expected valuation is less than $3\eps$, to aggressively bid at least $1-\eps$ in order to win bidder $2$. This greatly increases the revenue compared to the case of public signaling schemes in which case bidders bid only their expected values.

\begin{claim}\label{claim1:gap}
	The revenue obtained by the  optimal public signaling scheme in Example \ref{ex:AggreBid} is  at most $ 3\epsilon$. 
\end{claim}

\begin{claim}\label{claim2:gap}
	In Example \ref{ex:AggreBid}, there exists a private signaling scheme that achieves revenue at least  $1-\epsilon$ even in the worst Bayes Nash equilibrium.
\end{claim}

\begin{remark}Our earlier Example \ref{ex:curse} shows that a weak bidder may bid as low as possible if less informed, due to the winner's curse. Example \ref{ex:AggreBid} illustrates an interesting ``dual" phenomenon: a strong bidder may bid very aggressively at equilibrium if less informed. 
\end{remark}

\noindent \textbf{Discussion of the Bidding Behavior Assumption.} To obtain any non-trivial revenue guarantee in our worst case analysis, i.e., the worst Bayes Nash Equilibrium (BNE), it is \emph{necessary} to make some bidder behavior assumptions. Otherwise, the auctioneer's revenue in the \emph{worst} BNE of a second-price auction would always be 0 regardless of what signaling scheme she uses since that one bidder bids arbitrarily large and all other bidders bid $0$ is always a BNE. To rule out these trivial equilibria,  we assume throughout this section  that \emph{ any bidder will bid his true value if this is the dominant strategy} (also assumed in Proposition \ref{prop:gap}, Claim \ref{claim1:gap} and \ref{claim2:gap}).   Typically, the dominant strategy equilibrium is the most ideal solution concept to adopt for revenue analysis. Unfortunately, it does not necessarily exist in second price auctions under private signaling.  Our assumption can be viewed as a natural generalization of the dominant strategy equilibrium to the setting where it does not exist  --  i.e., bidders who have a dominant strategy will choose the dominant strategy and on top of that, we take the worst case analysis.

Proposition \ref{prop:gap}  shows the appealing revenue property of private signaling schemes, which motivates the computational question of designing the optimal private signaling scheme. Unfortunately, this turns out to be challenging due to several technical barriers.  The first is the problem of equilibrium selection. Recall that there are usually many equilibria under private signaling schemes, as illustrated in Example \ref{ex:curse}. Since there are no natural equilibrium selection criteria, we choose to analyze the revenue in the worst-case equilibrium, as indicated in Proposition \ref{prop:gap} (under the aforementioned bidding behavior assumption). 
Second, as illustrated by Example \ref{ex:curse} and \ref{ex:AggreBid}, private signaling schemes usually cause complicated bidding behaviors at equilibrium in second-price auctions. Even computing a desired equilibrium is already highly non-trivial \cite{Hausch1987,Syrgkanis2013}, let alone optimizing the revenue over these equilibria. 

Nevertheless, our next result shows that, under mild assumptions, there exists a private signaling scheme that extracts almost the \emph{full surplus} even in the \emph{worst} Bayes Nash equilibrium. Before stating the result formally, we first specify some notations. Let $\rho^* = \max_{\theta,i} v_i(\theta)$ denote the maximum value among all possible valuations of all bidders, and $i^*$ denotes the bidder  who has  value $\rho^*$ (one of them if multiple). Let $\rho^{**} = \max_{\theta, \, i \not = i^*} v_i(\theta)$ be the maximum value among all possible valuations, excluding the values of bidder $i^*$, and $i^{**}$ denote a bidder who has the value $\rho^{**}$. Let $x^+ = \max(x,0)$. We prove the following theorem.

\begin{theorem}\label{thm:private}
	Assume $0 \in \V_i$ for each $i$ and $\lambda$ supports on the \emph{entire} set  $\V$. There exists a private signaling scheme that can be implemented in $\poly(n,\sum_i |\V_i|)$ time  and extract the full surplus, excluding an amount of \, $\sum_{\theta} \lambda_{\theta}[v_{i^*}(\theta)-\rho^{**}]^+ + \epsilon$ for an arbitrarily small $\epsilon$, even in the worst Bayes Nash equilibrium.
\end{theorem}

The intuition underlying the missed amount from the surplus in Theorem \ref{thm:private} is as follows. Recall that $\rho^{**}$ is the largest possible value,  excluding the values of bidder $i^*$.  When $i^*$ has a value $v_{i^*}(\theta) > \rho^{**}$ at state of nature $\theta$, the revenue is at most $\rho^{**}$ in a second price auction since $\rho^{**}$ is the highest possible bid among all bidders other than $i^*$.  Therefore, an amount of $[v_{i^*}(\theta) - \rho^{**}]$  can \emph{not} be extracted from the surplus at such a state. The excluded amount in Theorem \ref{thm:private}, up to the negligible $\eps$, is precisely the sum of all such possible surplus losses weighted by their probabilities. In Example \ref{ex:AggreBid}, $\rho^{*} = \rho^{**} = 1$, therefore the event $v_{i^*}(\theta) > \rho^*$ does not happen and the excluded amount equals to $ \epsilon$, which is indeed the gap between the total surplus and the revenue of our constructed private signaling scheme (Claim \ref{claim2:gap}). 

The proof of Theorem \ref{thm:private} is inspired by Example \ref{ex:AggreBid}, which can be viewed differently from the following perspective.   In particular, the second value profile $(1,1-\epsilon,\eps)$ in Example \ref{ex:AggreBid} can be viewed as an \emph{auxiliary} value profile, which may have arbitrarily small probability but helps the auctioneer to extract almost full surplus from the first value profile $(2\epsilon, \epsilon,1)$.  
To prove Theorem \ref{thm:private}, we first generalize Example \ref{ex:AggreBid} by characterizing a general condition of two value profiles, under which an auxiliary value profile can be leveraged to extract (almost) full surplus from another value profile.  We then show that for any realized value profile, a corresponding auxiliary value profile can always be constructed under  the theorem assumptions. We defer the formal proof  to Appendix \ref{appendix:privateAlg}.

\section{Conclusions and Future Work}
In this paper, we study optimal signaling in second price auctions when there is a large set of possible item realizations. In a well-motivated Bayesian valuation setting, we obtain  preliminary positive and negative results for the optimal public signaling.  We also initiate the study of optimal private signaling in auctions and illustrate various challenges of designing the optimal private signaling scheme due to complicated equilibrium bidding behaviors and equilibrium selection issues. In the  basic known-valuation setting, we exhibit a private signaling scheme that extracts almost the full surplus. 

Our work has left open a rich set of interesting  questions. We conclude by mentioning a few. In the Bayesian valuation setting, we proved a menu-size-complexity-type negative result. However, the computational complexity of the problem is still open. In particular, is there a polynomial-time algorithm for optimal public signaling in the Bayesian valuation setting? Our positive result here relies on the MHR assumption of bidders' valuation distributions. Is it possible to relax this assumption and obtain an approximately optimal public scheme for more general settings? Is there still a ``simple" approximate scheme in the more general setting?

Our result for private signaling is still limited to the (less realistic) known valuation setting, but it  illustrates the surprising power of private schemes in extracting revenue. This opens many possibilities for future research. We wonder whether similar results hold in more realistic auction settings, e.g., the Bayesian valuation setting studied in this work. In particular, to what extent can private signaling schemes outperform public schemes in terms of extracting revenue? Is the auctioneer's revenue from the optimal private scheme close to the total surplus? How to efficiently compute the optimal private signaling scheme?  These descriptive and prescriptive questions are largely open  in more realistic auction settings.  



\bibliographystyle{named}
\bibliography{refer}

\appendix
\newpage

\DeclareRobustCommand*{\refa}{\ref{thm:public}}

\section{Proof of Theorem \refa}
\label{appendix:KnownPublic}
We start by exhibiting the linear program used by the algorithm, which uses the empirical distribution $\tilde{\lambda}$ and relaxes the first two constraints of LP \eqref{lp:public} by $\frac{\epsilon}{2n^2}$. 
\begin{figure}[H]
	\centering
	\begin{lp}\label{lp:public:empirical}
		\maxi{\sum_{\sigma^{ij} \in \Sigma}  \sum_{ \theta \in \tilde{\Theta}} \tilde{\lambda}_{\theta} \varphi(\theta,\sigma^{ij}) v_j(\theta) }
		\st 
		\qcon{\sum_{\theta \in \tilde{\Theta}} \tilde{\lambda}_{\theta} \varphi(\theta,\sigma^{ij}) v_i(\theta) \geq \sum_{\theta \in \tilde{\Theta}} \tilde{\lambda}_{\theta} \varphi(\theta,\sigma^{ij}) v_j(\theta) - \frac{\eps}{2n^2} }{j \not = i}
		\qcon{\sum_{\theta \in \tilde{\Theta}} \tilde{\lambda}_{\theta} \varphi(\theta,\sigma^{ij}) v_j(\theta) \geq \sum_{\theta \in \tilde{\Theta}} \tilde{\lambda}_{\theta} \varphi(\theta,\sigma^{ij}) v_k(\theta) - \frac{\eps}{2n^2}  }{k \not = i,j ; \, j \not = i}
		\qcon{\sum_{\sigma^{ij}} \varphi (\theta,\sigma^{ij}) = 1}{\theta \in \tilde{\Theta}}
		\qcon{\varphi (\theta,\sigma^{ij}) \geq 0}{\theta \in \tilde{\Theta}, \sigma^{ij} \in \Sigma}
	\end{lp} 	
	\caption{Relaxed Linear Program using Empirical Distributions}
\end{figure}

The algorithm simply solves Linear Program \eqref{lp:public:empirical} and signals for $\theta$ as suggested by the solution of LP \eqref{lp:public:empirical}. Details are in Algorithm~\ref{alg:public}, which we instantiate with $\eps >0$ and $K =\frac{8n^4}{\epsilon^2} \log \frac{4n^3}{\epsilon}$.  It is important to notice that the signal $\sigma^{ij}$ generated by Algorithm \ref{alg:public} does \emph{not} necessarily correspond to the outcome where bidder $i$ and $j$ have the highest and second-highest values. This is because we have relaxed the constraints in LP \eqref{lp:public:empirical} so that $i,j$ are not necessarily the top two bidders. Consequently, the objective of LP \eqref{lp:public:empirical} does \emph{not} correspond to the revenue of the auction. Nevertheless, Lemma \ref{lem:public:sum} shows that the revenue is close to the expected optimal objective value of LP \eqref{lp:public:empirical}, denoted as $OPT(LP\ref{lp:public:empirical})$. Theorem \ref{thm:public} follows from the following two lemmas.

\begin{algorithm}
	\begin{algorithmic}[1]
		\PARAMETER  $\eps \geq 0$; Integer $K \geq 0$.
		\INPUT State of nature $\theta$ represented by $v(\theta)$
		\INPUT Prior distribution $\lambda$ given as a sampling oracle 
		
		\OUTPUT Signal $\sigma \in \Sigma$, where $\Sigma= \set{\sigma^{ij}}_{i\not = j \in [n]}$.
		\STATE Draw integer $\ell$ uniformly at random from $\set{1,\ldots,K}$, and denote $\theta_{\ell} = \theta$.
		\STATE Sample $\theta_1, \ldots, \theta_{\ell-1}, \theta_{\ell+1}\ldots,\theta_K$ independently from the state of nature, and let $\tilde{\lambda}$ denote the empirical distribution over the multiset $\tilde{\Theta} = \set{\theta_1, \ldots, \theta_K}$.
		\label{step:sample}
		\STATE Solve linear program \eqref{lp:public:empirical}. Let $\tilde{\varphi}: \tilde{\Theta}  \to \Delta(\Sigma)$ be the optimal solution. \label{step:lpempirical}
		\STATE Output signal $\sigma^{ij}$ with probability $ \tilde{\varphi}(\theta_\ell, \sigma^{ij})$.
	\end{algorithmic}
	\caption{Public Signaling Scheme for Known-Valuation Setting}
	\label{alg:public}
\end{algorithm}

\begin{lemma}\label{lem:public:sum}
	Assume $\theta \sim \lambda$, and assume bidders bid their conditional expected value upon receiving a signal.  The expected  revenue  is at least $\Ex_{\tilde{\Theta}} OPT(LP\ref{lp:public:empirical}) - \frac{\eps}{2}$. Both expectations are taken over the random input $\theta$ as well as internal randomness and Monte-Carlo sampling performed by the algorithm.
\end{lemma}
\begin{proof}
	Let $\rev$ be the total expected revenue achieved by Algorithm \ref{alg:public}, $\rev(\sigma^{ij})$ denote the expected revenue generated by signal $\sigma^{ij}$, multiplied by the probability of receiving signal $\sigma^{ij}$. Therefore, $\rev = \sum_{\sigma^{ij} \in \Sigma} \rev(\sigma^{ij}) $. Let $v_k(\sigma^{ij})$ denote the bidder $k$'s value on receiving signal $\sigma^{ij}$, also multiplied by the probability of receiving signal $\sigma^{ij}$. Further, we denote $v_k(\sigma^{ij}| \tilde{\Theta}) = \sum_{\theta \in \tilde{\Theta}} \tilde{\lambda}_{\theta} \varphi(\theta,\sigma^{ij}) v_k(\theta)$. Therefore, $ v_k(\sigma^{ij}) = \Ex_{ \tilde{\Theta} } \big[  v_k(\sigma^{ij}| \tilde{\Theta}) \big]$. Note that, by the principle of deferred decision, $\tilde{\Theta} = \{ \theta_1,...,\theta_K \}$ at Step \ref{step:sample} can be viewed as $K$ i.i.d. samples from $\lambda$.   
	
	The first two constraints of LP \eqref{lp:public:empirical} mean $v_i(\sigma^{ij}| \tilde{\Theta}) \geq v_j(\sigma^{ij}| \tilde{\Theta}) - \frac{\eps}{2n^2}$ and $v_j(\sigma^{ij}| \tilde{\Theta}) \geq v_k(\sigma^{ij}| \tilde{\Theta}) - \frac{\eps}{2n^2}$ for all $k \not = i,j$.  Taking expectation over $\Theta$ on both sides of the inequalities, we have
	\begin{equation}\label{eq:signalRev}
		v_i(\sigma^{ij}) = \Ex_{ \tilde{\Theta} } \big[  v_i(\sigma^{ij}| \tilde{\Theta}) \big] \geq \Ex_{ \tilde{\Theta} } \big[  v_j(\sigma^{ij}| \tilde{\Theta}) - \frac{\eps}{2n^2} \big]  = v_j(\sigma^{ij}) - \frac{\eps}{2n^2}. 
	\end{equation}
	Similarly,  $ v_j(\sigma^{ij}) \geq  v_k(\sigma^{ij}) - \frac{\eps}{2n^2}$ for all $k \not = i,j$. Since the revenue of a second-price auction is at least the minimum of any two bidders' values, we have 
	\begin{equation} \label{eq:revBound}
		\rev(\sigma^{ij}) \geq  \min \bigg(  v_i(\sigma^{ij}), v_j(\sigma^{ij})  \bigg) \geq v_j(\sigma^{ij}) - \frac{\eps}{2n^2},
	\end{equation} 
	where the last inequality follows from Inequality \eqref{eq:signalRev}.  Therefore,
	\begin{align*}
		\rev & = \sum_{\sigma^{ij} \in \Sigma} \rev(\sigma^{ij}) \\
		&\geq  \sum_{\sigma^{ij} \in \Sigma} \bigg[ v_j(\sigma^{ij}) - \frac{\eps}{2n^2}  \bigg] \\
		&= \sum_{\sigma^{ij} \in \Sigma} \Ex_{ \tilde{\Theta} } \big[  v_j(\sigma^{ij}| \tilde{\Theta}) \big]  - \frac{\eps}{2} \\
		& =   \Ex_{ \tilde{\Theta} }  \bigg[ \sum_{\sigma^{ij} \in \Sigma}  \sum_{\theta \in \tilde{\Theta}} \tilde{\lambda}_{\theta} \varphi(\theta,\sigma^{ij}) v_j(\theta)  \bigg]  - \frac{\eps}{2} \\ 
		& =  \Ex_{ \tilde{\Theta} }  \, OPT(LP\ref{lp:public:empirical}) -  \frac{\eps}{2}  
	\end{align*}
	
\end{proof}

\begin{lemma}\label{lem:public:apx}
	Let  $OPT$ denote the expected revenue induced by the optimal  signaling scheme for prior distribution $\lambda$. When Algorithm \ref{alg:public} is instantiated with $K = \frac{8n^4}{\epsilon^2} \log \frac{4n^3}{\epsilon}$ and its input  $\theta$ is drawn from $\lambda$,  the expected revenue generated by Algorithm \ref{alg:public} is at least $OPT-\eps$.  Expectation is over the random input $\theta$ as well as the Monte-Carlo sampling performed by the algorithm.
\end{lemma}
\begin{proof}
	We prove that when $K = \frac{8n^4}{\epsilon^2} \log \frac{4n^3}{\epsilon}$, with probability at least $1 - \frac{\eps}{4}$, the expected optimal objective of LP \eqref{lp:public:empirical} is at least $OPT-\frac{\epsilon}{2}$. This, together with Lemma \ref{lem:public:sum}, implies that the revenue generated by Algorithm \ref{alg:public} is at least $OPT-\epsilon$.
	
	Let $\varphi^*$ be the solution of  LP \eqref{lp:public},  i.e., precisely the optimal signaling scheme of the problem.  At a high level,  we will show that restricting $\varphi^*$ to the set $\tilde{\Theta}$ results in a scheme that is feasible for LP \eqref{lp:public:empirical} and achieves revenue at least $OPT(LP\ref{lp:public:empirical})-\epsilon/2$ with high probability.  Formally, let $\tilde{\varphi}(\theta,\sigma) = \varphi^*(\theta,\sigma)$ for any $\theta \in \tilde{\Theta}, \, \sigma \in \Sigma$, and $\tilde{\lambda}$ be the empirical distribution over $\tilde{\Theta}$.  Observe that fixing any $i,j$, both the left-hand-side and right-hand-side terms in the first and second constraints of LP \eqref{lp:public} can be viewed as the expectation of random variable $\varphi^*(\theta,\sigma^{ij})$ over random  $\theta \sim \lambda$, while the corresponding constraints in LP \eqref{lp:public:empirical} with constructed solution $\tilde{\varphi}$ are  precisely their empirical means over $K$ samples. By the principle of deferred decision, these samples can be viewed as i.i.d. samples from $\theta \sim \lambda$. Therefore, with probability at least $1-e^{-2(\frac{\epsilon}{4n^2})^2 K} > 1 - \frac{\epsilon}{4n^3}$, the empirical mean of any term is within $\frac{\epsilon}{4n^2}$ of its expectation. Since there are $n$ different terms in total for any given $\sigma^{ij}$ and $|\Sigma| = n(n-1)$, by union bound, with probability at least $1 - \epsilon/4$, all the empirical means  are within $\frac{\epsilon}{4n^2}$ of their expectations.  In this case, since the original expectations satisfy the first two constraints of LP \eqref{lp:public}, the empirical means staisfy the first (relaxed) two constraints of LP \eqref{lp:public:empirical}, implying that $\tilde{\varphi}$ is feasible to LP \eqref{lp:public:empirical}. Moreover, the objective value of $\tilde{\varphi}$ in LP \eqref{lp:public:empirical} is  within  $\frac{\epsilon}{4n^2} \cdot n(n-1) \leq \epsilon/4$ of the optimal objective value of LP \eqref{lp:public}, i.e., $OPT$. Therefore, we have found a feasible $\tilde{\varphi}$ for LP \eqref{lp:public:empirical} that achieves objective at least $OPT -\frac{\eps}{4}$ with probability at least $1-\frac{\eps}{4}$. Randomness is over the samples $\tilde{\Theta}$. This implies that $\Ex_{\tilde{\Theta}}\big[ OPT(LP\ref{lp:public:empirical})\big] \geq OPT - \frac{\eps}{2}$, proving the lemma.  
	
\end{proof}

\newpage

\DeclareRobustCommand*{\refa}{\ref{thm:NumSignal}}

\section{Proof of Theorem \refa}
\label{appendix:NumSignalsA}

We construct a problem in which any scheme with $2^{o(n^{1/4})}$ signals must be far from optimality.  In fact, the problem we will construct is almost the simplest instance in the Bayesian valuation setting. 
In particular, consider a second-price auction in the Bayesian-valuation setting with $n$ i.i.d. bidders, described as follows. Any bidder $i$ has a valuation function $V$, which depends on her private type $t_i \in \{\mathbf{A},\mathbf{B}\}$ and a binary item type $\theta_i  \in \{0,1 \}$. Every bidder has the same valuation function $V(t_i,\theta_i)$, whose value table is as follows:   
\begin{table}[H]
	\centering
	\begin{tabular}{|c|c|c|}
		\hline
		& $V(\cdot,0)$ &  $V(\cdot,1)$\\ \hline
		$\mathbf{A}$ & $0$ & $1 $  \\  
		$\mathbf{B}$  & $0$ & $0 $  \\  \hline
	\end{tabular}
\end{table} 
$\theta_i$'s are i.i.d. with $\Pr(\theta_i = 1) = \eps$ for some small constant $\eps$. Therefore, the prior distribution of $\theta$ is a product distribution with $\Pr(\theta) = \epsilon^{|\theta|}(1-\epsilon)^{n - |\theta|}$. Moreover, bidders' private type $t_i$'s are also i.i.d. with $\Pr(t_i = \mathbf{A}) = \frac{1}{\sqrt{n}}$ for each $i$. 

For convenience, we will use $D_0$ to denote the point distribution at $0$ and $D_1$ to denote the  binary distribution which takes value  $1$ with probability $ \frac{1}{\sqrt{n}}$. Observe that in the constructed instance above, $\theta_i = 0$ or $1$ specifies bidder $i$'s valuation distribution is $D_0$ or $D_1$.  We start by showing that the revenue of the optimal public signaling scheme, denoted as $OPT_r$, is lower bounded by a  value arbitrarily close to $1$ for large enough $n$.


\begin{lemma}
	$OPT_r \geq 1 - e^{-\epsilon \sqrt{n}}  - \eps \sqrt{n} \cdot e^{-\epsilon \frac{n-1}{\sqrt{n}}}$. 
\end{lemma}
\begin{proof}
	We  consider the scheme of fully revealing $\theta$ to every bidder and show that the revenue of this particular public signaling scheme is already lower bounded by the RHS of the inequality in the lemma. Notice that $\Pr( |\theta| = i)  = C_n^i \epsilon^{|\theta|}(1-\epsilon)^{n - |\theta|}$. Moreover, when $|\theta| = i$, there  are $i$ bidders who have value distribution $D_1$. It is not hard to check that the revenue in this case is $1 - (1 - \frac{1}{\sqrt{n}})^i - \frac{i}{\sqrt{n}} (1 - \frac{1}{\sqrt{n}})^{(i-1)}$. Therefore, the revenue of full information scheme satisfies
	\begin{eqnarray*}
		\rev & = & \sum_{i=0}^{n} C_n^i \epsilon^{i}(1-\epsilon)^{n -i} \cdot \bigg[1 - (1 - \frac{1}{\sqrt{n}})^i  - \frac{i}{\sqrt{n}} (1 - \frac{1}{\sqrt{n}})^{(i-1)} \bigg]  \\
		& = & 1 - \sum_{i=0}^{n} C_n^i \epsilon^{i}(1-\epsilon)^{n - i} \cdot (1 - \frac{1}{\sqrt{n}})^i   - \sum_{i=0}^{n} C_n^i \epsilon^{i}(1-\epsilon)^{n -i} \cdot \frac{i}{\sqrt{n}} (1 - \frac{1}{\sqrt{n}})^{(i-1)}\\
		& = & 1 - \big( 1 - \epsilon + \epsilon - \frac{\epsilon}{\sqrt{n}} \big)^n - \eps \sqrt{n} \sum_{i=0}^{n} C_{n-1}^{i-1} (1-\epsilon)^{n -i} \cdot \epsilon^{i-1}(1 - \frac{1}{\sqrt{n}})^{(i-1)} \\ 
		& = & 1 - \big( 1 - \epsilon + \epsilon - \frac{\epsilon}{\sqrt{n}}  \big)^n - \eps \sqrt{n} \big( 1 - \epsilon + \epsilon - \frac{\epsilon}{\sqrt{n}}  \big)^{n-1} \\
		& \geq & 1 - e^{-\epsilon \sqrt{n}} - \eps \sqrt{n} \cdot  e^{-\frac{\epsilon(n-1)}{ \sqrt{n} } } .
	\end{eqnarray*}
	
	Therefore, $OPT_r \geq \rev \geq 1 - e^{-\epsilon \sqrt{n}} - \eps \sqrt{n} \cdot  e^{-\frac{\epsilon(n-1)}{ \sqrt{n} } } $, as desired. 
\end{proof}

We will now upper bound the revenue of any signaling scheme that uses $N$ signals, and show that the revenue is small when $N = 2^{o(n^1/4)}$.    Unfortunately, directly examining revenue is difficult. We instead consider an obvious upper bound of revenue, i.e., the welfare, and instead upper bound the maximum welfare of any signaling scheme using $N$ signals. 

We say a signaling scheme is \emph{deterministic} if it maps each state of nature deterministically  to one signal. As a result, each signal corresponds to a subset of states and all such subsets form a partition of the space of states of nature. A deterministic signaling scheme with $N$ signals is denoted by $\{S_1,...,S_N \}$, which is a partition of $\{0,1\}^n$.  We therefore also call $S_i$ a signal. The following lemma shows that there always exists a deterministic signaling scheme that achieves the maximum welfare. Intuitively, this is because welfare is a convex function (more precisely, the max function) of bidders' valuations, thus achieves its optimal value at a vertex of the polytope of signaling schemes. 
\begin{lemma}\cite{Dughmi2014}\label{lem:WelDeterministic}
	Let $\X_N $ denote all possible public signaling schemes using  at most  $N$ signals. There always exists a deterministic signaling scheme that achieves the maximum welfare among schemes in $\X_{N}$. Such a deterministic signaling scheme partitions the set $\{0,1\}^n$ into at most $N$ subsets.
\end{lemma}

As a result of Lemma \ref{lem:WelDeterministic}, it is without loss of generality to upper bound the welfare achieved by  deterministic signaling schemes, which are collections of subsets of $\{0,1 \}^n$. For any  $S \subseteq \{0,1\}^n$, we use $\Pr (S) = \sum_{\theta \in S} \lambda_{\theta}$ to denote the probability mass of the subset. At a high level, our idea is to show that any subset of $\{0,1\}^n$ either results in a welfare of $\O(\eps)$ or has  exponentially small probability mass (specifically, $\Pr (S) \leq 2^{-\Theta(n^{1/4})}$). This will imply that a deterministic scheme with  $2^{o(n^1/4)}$ signals (i.e., subsets) must have $\O (\eps)$ welfare. 

For any  $S \subseteq \{0,1\}^n$, let $$\bar{\theta}(S)  = \Ex (\theta| \theta \in S) = \frac{\sum_{\theta \in S} \epsilon^{|\theta|}(1-\epsilon)^{n-|\theta|} \cdot \theta}{\sum_{\theta \in S} \epsilon^{|\theta|}(1-\epsilon)^{n-|\theta|}}$$ denote the expectation of the states in $S$. When $S$ is clear from the context, we simply use $\bar{\theta}$ to denote the expected state. Observe that $\bar{\theta}_i$ is the probability that bidder $i$'s valuation function is $D_1$ conditioned on receiving the signal $S$.  The following lemma shows that if $\bar{\theta}(S)$ have many large components, then $\Pr (S)$ should be small.

\begin{lemma}\label{lem:BoundSetProb}
	For any $S \subseteq \{0,1 \}^n$, let $m$ be the number of entries in $\bar{\theta}(S)$ with values at least $2 \epsilon$. Then we have  $\Pr(S) \leq e^{-2\sqrt{m}}$. 
\end{lemma}
\begin{proof}
	By symmetry, w.l.o.g., we assume the first $m$ entries of $\bar{\theta}$ are at least $2 \epsilon$. We view this as a constrain on the set $S$ and will seek to maximize $\Pr(S)$ over all possible sets $S$ subject to the constraint. This can be abstractly formulated as an optimization problem.
	\begin{figure}[H]
		\centering
		\begin{lp}\label{opt:MaxP}
			\maxi{\Pr(S) = \sum_{\theta \in S} \epsilon^{|\theta|}(1-\epsilon)^{n-|\theta|}}
			\st
			\con{ S:\mbox{ values of first $m$ entries of $\bar{\theta}(S) \geq 2\epsilon$.}}
		\end{lp}%
		Program for Maximizing $\Pr(S)$
	\end{figure}   
	
	Unfortunately, Optimization Program \eqref{opt:MaxP} is difficult to reason about. We instead relax the constraint to be that the first $m$ entries sum up to at least $2m \epsilon$ and the consider the relaxed program as follows. 
	\begin{figure}[H]
		\centering
		\begin{lp}\label{opt:relaxMaxP}
			\maxi{\Pr(S) = \sum_{\theta \in S} \epsilon^{|\theta|}(1-\epsilon)^{n-|\theta|}}
			\st
			\con{\mbox{Sum of first $m$ entries of $\bar{\theta} \geq 2m\epsilon$.}}
		\end{lp}
		Relaxed Program for Maximizing $\Pr(S)$.
	\end{figure} 
	
	It turns out that the optimal solution to Optimization Program \eqref{opt:relaxMaxP} can be explicitly constructed. In particular, the program  seeks to maximize $\Pr(S)$ subject to that the sum of  first $m$ entries have large expected values. Some thoughts reveal that the optimal solution are obtained as follows: add all the $\theta$ with first $m$ entries summing up to $m$, then add all the $\theta$ with first $m$ entries summing up to $m-1$, and so on so forth, until that the expected sum of the first $m$ entries is less than $2m\eps$. Let $k$ denote the sum of the first $m$ entries of the last element added to $S$. Obviously, $k < 2m\eps$. We now argue that $k$ is at least $\eps m + m^{3/4}$. This requires a technical lemma which will be proved at the end of this section.
	\begin{lemma}\label{lem:CondExp}
		Let $X_1,X_2,...,X_m$ be $m$ i.i.d. Bernoulli random variables with $\Ex(X_i) = \eps$ for every $i$ and $X = \sum_{i=1}^m X_i$. 	Let $k = \eps m + m^{3/4} $ and $\eps \in (0, \frac{1}{2})$, there exists $N$ such that  $\Ex(X|X\geq k) = \frac{\sum_{i=k}^m C_m^i \eps^i(1-\eps)^{m-i} \cdot i }{ \sum_{i=k}^m C_m^i \eps^i(1-\eps)^{m-i}  } < 2m \eps$ for any $m \geq N$.  Here $C_m^i = \frac{m!}{i!(m-i)!}$ is the binomial coefficient.   
	\end{lemma}
	
	Lemma \ref{lem:CondExp} shows that $k$ is at least $\eps m + m^{3/4}$, since otherwise the expected sum of the first $m$ entries of $\bar{\theta}(S)$ is less than $2m \eps$, violating the constraint of Program \ref{opt:relaxMaxP}. As a result, $OPT(P\ref{opt:relaxMaxP})$ is upper bounded by $P(k) = \sum_{i=k}^m C_m^i \eps ^i(1-\eps)^{m-i} $, which is the probability that sum of $m$ i.i.d. binary random variables is at least $k$. By the Hoeffding bound, we know that  $P(k) \leq e^{-2\sqrt{m}}$ when $k \geq \eps m + m^{3/4}$. Therefore, $OPT(P \ref{opt:MaxP}) \leq OPT(\ref{opt:relaxMaxP}) \leq P(\eps m + m^{3/4}) \leq e^{-2\sqrt{m}}  $, concluding the proof. 
	
\end{proof}

Utilizing Lemma \ref{lem:BoundSetProb}, we can now formally show that for any signal $S \subseteq \{0,1\}^n$,  either the welfare from $S$ is  $\O(\eps)$ or its probability is exponentially small. 

\begin{lemma}\label{lem:BoundSetVal}
	For any $S \subseteq \{0,1\}^n$, at least one of the following two holds:
	\begin{enumerate}
		\item The conditional welfare of signal $S$ is at most $3\epsilon$. 
		\item $\Pr(S) \leq e^{-2n^{1/4}C}$ where $C =\sqrt{\ln \frac{1 - 2 \epsilon}{1-3\eps}} >0$ is a constant depending on $\eps$.  
	\end{enumerate} 
\end{lemma}
\begin{proof}
	Given $\bar{\theta}$,  the (expected) value distribution for bidder $i$ is $\bar{\theta}_i$ with probability $1/\sqrt{n}$ and $0$ otherwise.  Let $m$ denote the number of entries of $\bar{\theta}$ that are at least $2\epsilon$. By independence, the probability that at least one bidder's value is at least $2\epsilon$ is  $1 - (1 - \frac{1}{\sqrt{n}})^m$. 
	Therefore, the conditional welfare $\wel(S)$ is upper bounded $2 \epsilon \cdot (1 - \frac{1}{\sqrt{n}})^m + 1 - (1 - \frac{1}{\sqrt{n}})^m$. We show that if $\wel(S) > 3\eps$, then $m > \sqrt{n} \cdot \ln \frac{1 - 2 \epsilon}{1-3\eps} $, as follows
	\begin{eqnarray*}
		& & \wel(S) > 3 \epsilon \\
		&\Rightarrow & 2 \epsilon \cdot (1 - \frac{1}{\sqrt{n}})^m + 1 - (1 - \frac{1}{\sqrt{n}})^m > 3 \epsilon \\
		& \Rightarrow & 1 - 3 \epsilon > (1-2\eps) \cdot (1 - \frac{1}{\sqrt{n}})^m \\
		& \Rightarrow & \ln \frac{1 - 3 \epsilon}{1-2\eps} > m \ln  (1 - \frac{1}{\sqrt{n}}) \geq m \cdot ( - \frac{1}{\sqrt{n}})\\ 
		& \Rightarrow & m > \sqrt{n} \cdot \ln \frac{1 - 2 \epsilon}{1-3\eps} 
	\end{eqnarray*}
	Therefore, if $\wel(S) > 3 \epsilon$, we have by Lemma \ref{lem:BoundSetProb} that 
	$$\Pr(S) \leq e^{-2\sqrt{m}} \leq e^{-2n^{1/4}C} $$
	for $C=\sqrt{\ln \frac{1 - 2 \epsilon}{1-3\eps}}$. 
\end{proof}
Finally, let $\S = \{S_1,...,S_N \}$ be any deterministic signaling scheme with $N$ signals. Utilizing Lemma \ref{lem:BoundSetVal}, we have 
\begin{eqnarray*}
	\wel(\S) &=& \sum_{S: \, \wel(S) \leq 3 \eps} \wel(S) \cdot \Pr(S) + \sum_{S: \, \wel(S) > 3 \eps}  \wel(S) \cdot \Pr(S)  \\ 
	& \leq & \sum_{S: \, \wel(S) \leq 3 \eps} 3\eps  \cdot \Pr(S) + \sum_{S: \, \wel(S) > 3 \eps}  \wel(S) \cdot  e^{-2n^{1/4}C} \\ 
	& \leq & 3 \eps + N\cdot e^{-2n^{1/4}C},
\end{eqnarray*}
where constant $C =\sqrt{\ln \frac{1 - 2 \epsilon}{1-3\eps}}$. For any $N = 2^{o(n^{1/4})}$, $3 \eps + N\cdot e^{-2n^{1/4}C}$ is upper bounded by $4\eps$ for large enough $n$. On the other hand, the revenue is close to $1$. Since $\eps$ can be an arbitrarily small constant, any signaling scheme with $2^{o(n^{1/4})}$ signals cannot guarantee a constant fraction of the optimal revenue for any constant less than $1$. This concludes our proof of Theorem \ref{thm:NumSignal}. 

\DeclareRobustCommand*{\refa}{\ref{lem:CondExp}}
\subsection*{Proof of the Technical Lemma \refa}
(Note: we use $p$ instead of $\eps$ as in the original lemma statement.)

\noindent {\it Lemma Statement:
	Let $X_1,X_2,...,X_m$ be $m$ i.i.d. Bernoulli random variables with $\Ex(X_i) = p$ for every $i$ and $X = \sum_{i=1}^m X_i$. 	Let $k = p m + m^{3/4} $ and $p \in (0, \frac{1}{2})$, there exists $N$ such that  $\Ex(X|X\geq k) = \frac{\sum_{i=k}^m C_m^i p^i(1-p)^{m-i} \cdot i }{ \sum_{i=k}^m C_m^i p^i(1-p)^{m-i}  } < 2m p$ for any $m \geq N$, where $C_m^i$ is the binomial coefficient.   
}

For convenience, we define the following terms:
\begin{eqnarray*}
	C(m,k;p) &=& \Pr(X=k) = C_m^k p^k(1-p)^{m-k}. \\
	S(m,k;p) &=& \Pr(X \geq k) = \sum_{i=k}^m C_m^i p^i(1-p)^{m-i} \\
	E(m,k;p) &= &  \Ex (X| X \geq k) \cdot \Pr(X \geq k)=  \sum_{i=k}^m C_m^i p^i(1-p)^{m-i} \cdot i 
\end{eqnarray*} 
\begin{lemma}\label{lem:boundSC}
	Fix $p$, there exists large enough $N$ such that for any integer $m \geq N$ and  $k \in  [pm, pm+m^{3/4} + 1]$, we have $$S(m,k+1;p) > C(m,k;p).$$
\end{lemma}
\begin{proof}
	We first upper bound the ratio between $C(m,k+1;p)$ and $C(m,k;p)$ for $k \in  [pm, pm+m^{3/4} + 3]$.
	\begin{eqnarray*}
		r_k &=& \frac{C(m,k+1;p)}{C(m,k;p)} \\
		& = & \frac{C_m^{k+1} p^{k+1}(1-p)^{m-k-1}}{C_m^k p^k(1-p)^{m-k}} \\
		& = & \frac{m-k} {k+1}\cdot \frac{p}{1-p} \\
		& \geq &  \frac{m-pm - m^{3/4}-3}{pm + m^{3/4}+4} \cdot \frac{p}{1-p} 
	\end{eqnarray*}
	Note that $\frac{m-pm - m^{3/4}-3}{pm + m^{3/4}+4} \leq \frac{m-pm }{pm } =  \frac{1-p}{p}$, and tends to $\frac{1-p}{p}$ as $m \to \infty$. Therefore, there exists large enough $N$ such that $\frac{n-pm - m^{3/4}-3}{pm + m^{3/4}+4} \geq \frac{3}{4} \cdot \frac{1-p}{p}$, thus $r_k \geq \frac{3}{4}$,  for any $m \geq N$. Therefore, for any $k \in  [pm, pm+m^{3/4}+1]$.
	\begin{eqnarray*}
		S(m,k+1;p) &>& C(m,k+1;p) + C(m,k+2;p) \\
		& = & r_k C(m,k;p) + r_k r_{k+1} C(m,k;p) \\
		& \geq & \bigg[ \frac{3}{4} + (\frac{3}{4})^2 \bigg] \cdot  C(m,k;p) \\
		& \geq & C(m,k;p).
	\end{eqnarray*}
	This shows that $S(m,k+1;p) > C(m,k;p)$ for any $m \geq  N$, as desired. 
\end{proof}

{\bf Proof of the Lemma:}

Note that $\Ex (X| X \geq k) = \frac{E(m,k;p)}{S(m,k;p)}$. 	Let $\Delta(m,k;p)= E(m,k;p) - 2pm S(m,k;p)$. We prove that $\Delta(m,k;p)< 0$. 
\begin{eqnarray*}
	\frac{\Delta(m,k;p)}{m} &=&  \frac{1}{m}\sum_{i=k}^m i \cdot C_m^i p^i(1-p)^{m-i}   - 2p \cdot  \sum_{i=k}^mC_m^i p^i(1-p)^{m-i} \\
	& = & \sum_{i=k}^m C_{m-1}^{i-1} p^i(1-p)^{m-i}  - 2p \cdot \sum_{i=k}^m \bigg[C_{m-1}^{i-1} + C_{m-1}^i\bigg] p^i(1-p)^{m-i} \\
	& = & (1-2p) \cdot \sum_{i=k}^m C_{m-1}^{i-1} p^i(1-p)^{m-i}  - 2p \cdot  \sum_{i=k}^m  C_{m-1}^ip^i(1-p)^{m-i} \\
	& = & (1-2p) \cdot \sum_{i=k-1}^{m-1} C_{m-1}^{i} p^{i+1}(1-p)^{m-i-1}  - 2p \cdot  \sum_{i=k}^{m-1}  C_{m-1}^ip^i(1-p)^{m-i} \\
	& = & (1-2p)p \cdot \sum_{i=k-1}^{m-1} C_{m-1}^{i} p^{i}(1-p)^{m-1-i}  - 2p(1-p) \cdot  \sum_{i=k}^{m-1}  C_{m-1}^ip^i(1-p)^{m-1-i} \\
	& = & (1-2p)p \cdot \bigg[ S(m-1,k;p) + C(m-1,k-1;p) \bigg]  - 2p(1-p) \cdot  S(m-1,k;p) \\
\end{eqnarray*}

Note that $m^{3/4} < (m-1)^{3/4} + 1$ for any $m> 1$, therefor $k-1 = pm -1 + m^{3/4} < p(m-1) + (m-1)^{3/4} + 1$. It is easy to see that $k-1 > p(m-1)$. As a result, $k-1 \in [p(m-1), p(m-1) + (m-1)^{3/4} + 1]$.  By Lemma \ref{lem:boundSC}, there exists $N$ such that $S(m-1,k;p) > C(m-1,k-1;p)$ for any $m \geq N$. In this case, we have 
$$\frac{\Delta(m,k;p)}{m} <  S(m-1,k;p) \bigg[ (1-2p)2p - 2p(1-p) \bigg] \leq 0. $$
Therefore, $\Delta(m,k:p) < 0$ for any $m \geq N$. 

\newpage

\DeclareRobustCommand*{\refa}{\ref{sec:Bayes}}

\section{Omitted Proofs from Section \refa}
\DeclareRobustCommand*{\refa}{\ref{lem:tail_balance}}

\subsection{Proof of Lemma \refa } \label{appendix:BayesAlgLem}
Since any $\theta \in \Theta_1$ can be viewed as a bidder $i$, we will think of $\Theta_1$ as $[n]$ and a bidder $i$ as a state in $\Theta_1 = [n]$. Let $p_i$ denote the probability of state $i$.   W.l.o.g., assume that $p_1 \geq p_2 ... \geq p_n$. We now prove each direction. 

`` $\Rightarrow$": Since $p_1 \cdot \pi(1,i) =  p_i \cdot \pi(i,1)$ for any $i = 1,...n$ (note $\pi(1,1) = 0$), summing over $i$, we obtain $$p_1 = \sum_{i=1}^n p_1 \cdot \pi(1,i) = \sum_{i=1}^n p_i \cdot \pi(i,1) =\sum_{i=2}^n p_i \cdot \pi(i,1)  \leq \sum_{i=2}^n p_i. $$
Since $p_1 \geq p_j$ for any $j$, we have $p_j \leq \sum_{i\not = j} p_i$, as desired. 

`` $\Leftarrow$": We show that a pooling scheme can be explicitly constructed. Observe that the pooling scheme can be viewed as matching bidders with equal probability mass. Particularly, bidder $i$ and $j$ are matched with probability mass $\lambda_i \pi(i,j)$. We first show that if $p_1 = p_2$, there always exists a feasible pooling scheme. This is because we can match bidder $n$ with $n-1$ for probability mass of $p_n$ so that the left probability mass of bidder $n-1$ is $p'_{n-1} = p_{n-1} - p_n$ while bidder $n$ has $0$ probability mass left. We then match bidder $n-1$ with bidder $n-2$ for probability mass $p'_{n-1}$, etc. This process will continue until only bidder $1,2,3$ have non-zero probability masses left, satisfying $p_1 = p_2 \geq p'_3$. Finally, we can match bidder $1$ with bidder $3$ for probability mass $p'_3/2$,  bidder $2$ with bidder $3$ for probability mass $p'_3/2$, and bidder $1$ with bidder $2$ for probability mass $p_1 - p'_3/2$, yielding a feasible pooling scheme.  

When $p_1 > p_2$, we can reduce it to the case with $p_1 = p_2$, as follows. We take probability mass of $p_1 - p_2$  from bidder $1$, and use it to  match bidder $1$ with all other bidders except bidder   $2$, i.e., match bidder $1$ with bidder $3$, bidder $4$, ..., bidder $n$.  Since $p_1 \leq \sum_{i=2}^n p_i$, i.e., $p_1 - p_2 \leq \sum_{i=3}^n p_i$, the probability mass of $p_1 - p_2$ from bidder $1$ will be used up to match bidder $1$ with other bidders except bidder $2$, at which point we arrive at an instance satisfying $p_1 = p_2$.   A pooling scheme for the left probability masses can then be constructed using the earlier procedure.   


\DeclareRobustCommand*{\refa}{\ref{thm:unknownAlgo}}
\DeclareRobustCommand*{\refaa}{\ref{lem:unknownAlgo}}

\subsection{Proof of Theorem \refa \, and Lemma \refaa }
\label{appendix:BayeAlgo}

Recall that the welfare is maximized by revealing full information, and the optimal welfare is the expectation of the welfare at each state of nature. Our proof will lower bound the revenue in terms of the optimal welfare. In fact, we will prove an even stronger result: given any realized state of nature, the revenue generated at that state is at least $1/8$ fraction of the welfare at that state. We prove separately for the states that are fully revealed and these that are pooled. 

Our analysis involves various comparisons of welfare and revenue in the standard second price auction (with no signaling).  The main technical barrier comes from the fact that bidders are \emph{asymmetric} in the sense that they have either the high or low distribution at any given state of nature. This makes it difficult to analytically characterize  revenue and welfare. Our strategy is to first analyze the \emph{symmetric} case in which all bidders have the same value distribution $H(v)$, and then extend the bound to asymmetric bidders. We use $A_k$ to denote a standard second price auction with $k$ i.i.d. bidders, and  $\rev(A_k;H)$ and $\wel(A_k;H)$ denote the revenue and welfare respectively when bidders in $A_k$ have value distribution $H(v)$. Let $A_{n,k}$ denote a standard second price auction with $n$ independent bidders, among whom $k$ bidders have value high distribution $H(v)$ and another $n-k$ bidders have low value distribution  $L(v)$. $\rev(A_{n,k})$ and $\wel(A_{n,k})$ denote the revenue and welfare of the auction $A_{n,k}$, respectively. The notation $v_{[1]}$ and $v_{[2]}$ are reserved for the largest and second largest bidder value in the auction.

In auction $A_k$ with distribution $H(v)$, $\Pr(v_{[1]} \leq v) = H^k(v)$ and $\Pr(v_{[2]} \leq v) = k H^{k-1}(v) - (k-1)H^k(v)$. The welfare and revenue is the expectation of $v_{[1]}$ and $v_{[2]}$ respectively, which can be expressed as follows.  Here, we use constant $C_k = 1 + \frac{1}{2} + ... \frac{1}{k}$ denote the harmonic number.

\begin{eqnarray*}
	& & \mathbf{Wel}(A_k;H) \\
	&=& \int_{0}^{\infty} [1 - H^k(v)] dv \\
	&=& \int_{0}^{\infty} \frac{1 - H(v)}{h(v)} \bigg[1 + H(v) ... + H^{k-1}(v) \bigg] h(v) dv\\
	&=&  \frac{1 - H(v)]}{h(v)} \bigg[H(v) + \frac{H^2(v)}{2} ... + \frac{H^{k}(v)}{k}\bigg]\bigg|_0^{\infty}  - \int_{0}^{\infty} \bigg( \frac{1 - H(v)]}{h(v)} \bigg)' \bigg[H(v) + \frac{H^2(v)}{2} ... + \frac{H^{k}(v)}{k} \bigg]dv \\
	&=&  \lim_{v \to \infty}\frac{1 - H(v)}{h(v)} \cdot C_k  + \int_{0}^{\infty} \bigg(-\frac{1 - H(v)]}{h(v)} \bigg)' \bigg[H(v) + \frac{H^2(v)}{2} ... + \frac{H^{k}(v)}{k} \bigg]dv \\
\end{eqnarray*}
and 
\begin{eqnarray*}
	& & \mathbf{Rev}(A_k;H) \\
	& = & \int_{0}^{\infty} [1 - k H^{k-1}(v) + (k-1)H^k(v)] dv \\
	&=& \int_{0}^{\infty } \frac{1 - H(v)}{h(v)} \bigg[1 + H(v) ... + H^{k-2}(v) - (k-1)H^{k-1}(v) \bigg] h(v) dv\\
	&=&  \lim_{v \to \infty}\frac{1 - H(v)}{h(v)} \cdot (C_k - 1) + \int_{0}^{\infty} \bigg(-\frac{1 - H(v)]}{h(v)} \bigg)' \bigg[H(v) +  \frac{H^2(v)}{2} ... + \frac{H^{k-1}(v)}{k-1} -  \frac{(k-1)H^{k}(v)}{k} \bigg]dv \\
\end{eqnarray*}

The last equality in both expressions comes from integration by parts. For convenience, let $x = H(v) \in [0,1]$. The term in the square bracket of $\rev(A_k;H)$ is denoted as $R_k(x) = x + x^2/2 + ..., x^k/k - x^k $, while the term in the square bracket of $\wel(A_k;H)$ is denoted as $W_k(x) = x + x^2/2 +..., x^k/k$. The following lemma is central to our analysis. 
\begin{lemma}\label{lem:BayesAlgo:bound1}
	For any integer $k,i>0$, if there exists $\rho > 0$ such that 
	$R_k(x) \geq \rho \cdot W_i(x)$ for all $x \in [0,1]$, then $\rev(A_k;H) \geq \rho \cdot  \wel(A_i;H)$. 
\end{lemma}
\begin{proof}
	Notice that  $\frac{1-H(v)}{h(v)} \geq 0$ by definition, and $\bigg(-\frac{1 - H(v)]}{h(v)} \bigg)' \geq 0$ by the MHR assumption. Comparing the terms in $\rev(A_k;H)$ and $\wel(A_i;H)$, it is not hard to conclude that the assumed conditions are sufficient to guarantee $\rev(A_k;H) \geq \rho \cdot  \wel(A_i;H)$ (note that $C_k - 1 = R_k(1)$ and $C_k = W_k(1)$). 
\end{proof}
Using Lemma \ref{lem:BayesAlgo:bound1}, we now prove that when $|\theta|_1 \geq 2$ or $|\theta|_1 = 0$, in which case we fully reveal $\theta$ and then run a standard second price auction, the revenue is at least $\frac{1}{8}$ fraction of welfare.
\begin{lemma}\label{lem:BayesAlgo:bound2}
	$\rev(A_k;H) \geq \max\{ \frac{1}{3}, \frac{\ln (k+1) - 1}{\ln (k+1)} \}  \cdot  \wel(A_k;H)$ for any $k \geq 2$. 
\end{lemma}
\begin{proof}
	It is easy to check   $$f(x) = \frac{R_k(x)}{W_k(x)} = 1 - \frac{x^k}{x+x^2/2 + ... x^k/k}$$ is a  decreasing function for $x \in [0,1]$, therefore, we have  
	$$ f(x) \geq f(1) = \frac{C_k - 1}{C_k}$$
	That is, $R_k(x) \geq \frac{C_k - 1}{C_k}W_k(x)$. Invoking Lemma \ref{lem:BayesAlgo:bound1}, we have 
	$$ \frac{\mathbf{Rev}(A_k; H)}{\wel(A_k; H)} \geq \frac{C_k-1}{C_k} \geq  \frac{ 1}{3} $$ where we used the inequality $C_k=1 + \frac{1}{2} ... + \frac{1}{k} \geq 3/2$ when $k \geq 2$ and $\frac{x-1}{x}$ is increasing in $x$.  Similarly, $\mathbf{Rev}(A_k; H) \geq \frac{\ln (k+1) - 1}{\ln (k+1)} \wel(A_k; H)$ since $C_k \geq \ln(k+1)$. This proves the lemma. 
\end{proof}

Lemma \ref{lem:BayesAlgo:bound2} implies that when $|\theta| = 0$, in which case we reveal full information, the revenue is at least $\max\{ \frac{1}{3}, \frac{\ln (k+1) - 1}{\ln (k+1)} \}$ fraction of the welfare (Part 1 of Lemma \refaa ). 
We now prove similar guarantee for the case $|\theta| \geq 2$. This is proven by extending the above bound to the case with asymmetric bidders.

\begin{lemma}\label{lem:BayesAlgo:full}[Part 2 of Lemma \refaa]
	For $k \geq 2$,	$\rev(A_{n,k}) \geq \max\{ \frac{1}{6}, \frac{\ln (k+1) - 1}{2\ln (k+1)} \} \cdot \wel(A_{n,k})$.
\end{lemma}
\begin{proof}
	Let $m = n-k$ denote the number of low-distribution bidders, $\bv{u} \in \R^k, \bv{v} \in \R^m$ denote the value profile of the high-distribution bidders and low-distribution bidders respectively. Notice that $\max (\bv{u} ,\bv{v}) \leq \max (\bv{u}) + \max (\bv{v})$ for any $\bv{u} \in \R^k, \bv{v} \in \R^m$,  we therefore have 
	\begin{eqnarray*}
		\wel(A_{n,k}) &=& \Ex_{\bv{u}  \sim H^k} \Ex_{\bv{v} \sim L^m} \max (\bv{u} ,\bv{v}) \\ 
		& \leq & \Ex_{\bv{u}  \sim H^k} \Ex_{\bv{v} \sim L^m} [ \max (\bv{u}) + \max(\bv{v})  ] \\
		& =& \mathbf{Wel}(A_k; H) + \mathbf{Wel}(A_m; L)
	\end{eqnarray*}
	We now lower bound $\rev(A_{n,k})$ in terms of $\wel(A_{n,k})$. Since the auction consists of $k$ high-distribution bidders and  $m$ low-distribution bidders, we have $\rev(A_{n,k}) \geq \rev(A_k;H)$ and $\rev(A_{n,k}) \geq \rev(A_m;L)$. 
When $k \leq m$, we have 
	\begin{eqnarray*}
		\rev(A_{n,k}) &\geq& \frac{1}{2} \big( \mathbf{Rev}(A_m; L) + \mathbf{Rev}(A_k; H) \big) \\
		&\geq &   \frac{1}{2} [ \max\{ \frac{1}{3}, \frac{\ln (m+1) - 1}{\ln (m+1)} \} \cdot \mathbf{Wel}(A_m; L) +  \max\{ \frac{1}{3}, \frac{\ln (k+1) - 1}{\ln (k+1)} \}  \cdot \mathbf{Wel}(A_k; H)] \\
		&\geq&  \max\{ \frac{1}{6}, \frac{\ln (k+1) - 1}{2\ln (k+1)} \} \mathbf{Wel}(A_{n,k}).
	\end{eqnarray*} 
	where the second inequality follows from Lemma \ref{lem:BayesAlgo:bound2}. 
	
	When $k > m$ or equivalently $k > \frac{1}{2}n$, we apply Lemma \ref{lem:BayesAlgo:bound1} to get $\mathbf{Rev}(A_k; H) \geq \max\{ \frac{1}{6}, \frac{\ln (k+1) - 1}{2 \ln (k+1)} \} \mathbf{Wel}(A_n; H)$. To do so, we lower bound the following function: 
	\begin{eqnarray*}
	f(x) &=& \frac{R_{k}(x)}{W_n(x)} \\
	&=& \frac{x+x^2/2 + \cdots + x^{k}/k - x^{k}}{x+x^2/2 + \cdots + x^n/n} \\
	& = & 1 - \sum_{i = k+1}^n \frac{x^{i}/i}{x+x^2/2 + \cdots + x^n/n}  - \frac{ x^{k}}{x+x^2/2 + \cdots +x^n/n} \\
	& \geq &  1 - \sum_{i = k+1}^n \frac{x^{i}/i}{x+x^2/2 + \cdots + x^i/i}  - \frac{ x^{k}}{x+x^2/2 + \cdots +x^k/k} \\
	& \geq &  1 - \sum_{i = k+1}^n \frac{1}{ i \cdot C_i}  - \frac{ 1}{C_k}
	\end{eqnarray*}
where the last inequality uses the fact that $ \frac{x^{i}}{x+x^2/2 + \cdots + x^i/i}$ achieves maximum at $x = 1$. As a result, $f(x)$ can be further lower bounded as follows:
$$ f(x) \geq 1 - \frac{1}{C_k} \cdot \bigg( 1 + \sum_{i = k+1}^n \frac{1}{ i}   \bigg) \geq 1 - \frac{1+\ln(n-1) - \ln(k)}{\ln(k+1)} \geq  \max\{ \frac{1}{6}, \frac{\ln (k+1) - 1}{2\ln (k+1)} \},$$
where the second inequality is due to the fact $\sum_{i = k+1}^n \frac{1}{ i}   \leq \ln(n-1) - \ln(k)$ and the last inequality is easy to verify to be true for any $k \geq 12$ (since $n \geq 22$ by assumption of Theorem \ref{thm:unknownAlgo}). 

As a result, when $k \geq m$, we have $$ \mathbf{Rev}(A_{n,k}) \geq  \max\{ \frac{1}{6}, \frac{\ln (k+1) - 1}{2 \ln (k+1)} \} \wel (A_{n},H) \geq \max\{ \frac{1}{6}, \frac{\ln (k+1) - 1}{2 \ln (k+1)} \} \wel (A_{n,k}).$$ 
\end{proof}

So far we have shown that at the states that are fully revealed, the revenue is at least $\frac{1}{6}$ fraction of the welfare. We now prove similar guarantees for the states that are pooled. Note that when Algorithm \ref{alg:publicBayes} outputs a pool signal $(\theta',pool)$ where $\theta' \in \Theta_2$, the two bidders $i,j \in \theta'$ are equally likely (by definition of the $1$-tail balanced pooling scheme), therefore each  bidder in $\theta'$ has a high distribution with probability $1/2$. Assume bidders bid in expectation, the revenue in this case is at least $\frac{1}{2} \cdot \rev(A_{n,2})$. The following lemma shows that $\frac{1}{2}\rev(A_{n,2}) \geq \frac{1}{8} \wel(A_{n,1})$. This implies that in Algorithm \ref{alg:publicBayes},  the revenue at any pooled state is at least $\frac{1}{8}$ fraction of the welfare of that state. 
\begin{lemma}\label{lem:BayesAlgo:pool}[Part 3 of Lemma \refaa]
	When $n \geq 22$, we have	$\frac{1}{2} \cdot \rev(A_{n,2}) \geq \frac{1}{8} \cdot \wel(A_{n,1})$.
\end{lemma}
\begin{proof}
	We first prove $\rev(A_2; H) \geq \frac{1}{2} \wel(A_1; H)$.  This is because $C_2 -1 = \frac{1}{2} = \frac{1}{2}C_1$ and  $\frac{R_2(x)}{ W_1(x)} =  \frac{x  - \frac{x^2}{2} }{ x} \geq \frac{1}{2}$ for any $x \in [0,1]$.  Invoking Lemma \ref{lem:BayesAlgo:bound1}, we obtain the conclusion.

	Now we argue that 	when $n \geq 22$, $\rev(A_{n-2};H) \geq \frac{1}{2}\wel(A_{n-1}; H) $. Note that $C_{n-1}  \leq 2[C_{n-2}-1]$ for any $n \geq 22$. By Lemma \ref{lem:BayesAlgo:bound1}, we only need to prove $W_{n-1}(x) \leq 2R_{n-2}(x)$ for any $x \in [0,1]$ and $n \geq 22$, as follows:
	\begin{eqnarray*}
		\frac{R_{n-2}(x)}{W_{n-1}(x)} &=& \frac{x + x/2 + ..., x^{n-2}/(n-2) - x^{n-2} }{ x + x/2 +..., x^{n-1}/(n-1) } \\
		&=& 1 - \frac{x^{n-1}/(n-1) +  x^{n-2} }{ x + x/2 +..., x^{n-1}/(n-1) } \\
		& \geq & 1 - \frac{(\frac{1}{n-1}  + 1)x^{n-2} }{ x + x/2 +..., x^{n-2}/(n-2) } \\
		& \geq & \frac{C_{n-2} - 1 -\frac{1}{n-1}  }{C_{n-2}} \\
		& \geq & \frac{1}{2}
	\end{eqnarray*}

	Similar to the proof in Lemma \ref{lem:BayesAlgo:full},  we have $\wel(A_{n,1}) \leq \wel(A_1; H) + \wel(A_{n-1}; L)$ and 
	$\rev (A_{n,2}) \geq \frac{1}{2}(\rev(A_2; H) + \rev(A_{n-2}; L))$.
	Therefore 
	$$
	\rev (A_{n,2})	\geq    \frac{1}{2} \cdot  [ \frac{1}{2} \wel(A_1; H) + \frac{1}{2}  \wel(A_{n-1}; L)] 
	\geq  \frac{1}{4} \wel(A_{n,1}). 
	$$
\end{proof}

\newpage

\DeclareRobustCommand*{\refa}{\ref{sec:PublicPrivate}}

\section{Omitted Proofs from Section \refa}
\DeclareRobustCommand*{\refa}{\ref{claim1:gap}}
\DeclareRobustCommand*{\refaa}{\ref{claim2:gap}}
\subsection{Proof of Claim \refa \, and Claim \refaa}

{\bf Proof of Claim \refa}. Note that in public signaling schemes, bidders always bid their expected values by our behavior assumption. In Example \ref{ex:AggreBid}, bidder $1$'s valuation is higher than bidder $2$'s at any state, therefore bidder $1$'s bid is always greater than bidder $2$'s, conditioned on any signal. This implies that bidder $1$'s bid is either the highest or the second highest given any signal. As a result, the revenue is at most bidder $1$'s bid, which in expectation is less than $3\eps$.

\vspace{3mm}
\noindent {\bf Proof of Claim \refaa}. Consider the following private signaling scheme: the auctioneer informs bidder 2 and 3 the true state  but reveals no information to bidder 1.  Since bidder 2 and 3 know the exact state, bidding the true value is their dominant strategy.  We claim that the equilibrium bid of bidder $1$ lies in the interval $[1 - \epsilon,1)$.  To see this, observe that any bid less than $1-\epsilon$ does not place bidder 1 in the winning position in neither state, thus results in expected utility $0$ to bidder 1. However, any bid  $ b \in (1-\epsilon, 1)$ helps bidder $1$ to win the auction at the second state, which results in a strictly positive utility for bidder 1.  Therefore, bidder 1 bids at least $1-\eps$ in any equilibrium,  and  the auctioneer's revenue is also at least $1-\eps$, concluding the proof.

\DeclareRobustCommand*{\refa}{\ref{thm:private}}
\subsection{Proof of Theorem \refa}
\label{appendix:privateAlg}

Example \ref{ex:AggreBid} illustrates a situation in which the auctioneer can use a private signaling scheme to extract revenue of at least $1-\eps$, which is close to the total surplus $1$. Interestingly,  this example can be thought differently from the perspective of signaling schemes. In particular, imagine we are given any realized value profile $(2\epsilon, \epsilon, 1)^T$, Example \ref{ex:AggreBid} and Claim \ref{claim2:gap} show that  by ``mixing" this value profile with another value profile $(1,1-\epsilon,\eps)^T$ of  arbitrarily small probability,  we can extract almost full surplus from the given value profile $(2\epsilon, \epsilon, 1)^T$ using private signaling schemes. If we can do such mixing for any realized value profile, the aggregated revenue should be close to the full surplus. It turns out that this insight can be made precise.

We start by generalizing Example \ref{ex:AggreBid} and characterizing the structure of value profiles that  leads to high revenue.  Throughout the proof,  we will use $u_i$ to denote the $i$'th entry of any vector $u$ and $u_{[i]}$ to denote the $i$'th largest entry of $u$. The following lemma plays a key role in our poof. 
\begin{lemma}\label{lem:private}
	Consider a second-price auction in the known-valuation setting with two possible value profiles $v,u \in \RR^n$ that have the following structure: 
	\begin{table}[H]
		\centering
		\begin{tabular}{|c|c|c|c|c|}
			\hline
			\multirow{2}{*}{} & \multicolumn{2}{c|}{Probability} \\
			\hhline{~----}
			& $v: \mbox{ prob. }(1-\delta)$ & $u: \mbox{ prob. }\delta$ \\
			\hline
			...  & ...      & ...  \\
			bidder $j$  & ...        &  $u_{[1]}\, (>u_{[2]})$ \\ 
			...  & ...      & ... \\
			bidder $k$  & ...                  & $u_{[2]} \,(= v_{[1]} - \eps) $   \\
			...  & ...                 & ... \\
			bidder $i$  & $v_{[1]}$                  & ... \\
			...  & ...                  & ... \\
			\hline
		\end{tabular}
	\end{table}
	
	If $j \not = i$ (but $k$ could equal $i$) and $v_{[1]},u_{[1]}, u_{[2]}$ satisfy the conditions in the above table , then for any $\delta >0$, the following private signaling scheme extracts revenue at least $(v_{[1]} - \epsilon)$ even in the worst Bayes Nash equilibrium: reveal no information to bidder $j$ and reveal full information to every other bidder.  
	
\end{lemma}
The proof of Lemma \ref{lem:private} follows a similar argument as that of Claim \ref{claim2:gap}, thus omitted here.  It is important to notice that Lemma \ref{lem:private} does not depend on the probability $\delta$ so long as it is strictly positive. Therefore, given any realized value profile v, if we can  mix it with a value profile $u$ of arbitrarily small probability that satisfied the properties in Lemma \ref{lem:private}, then the corresponding private scheme will force bidder $j$ -- the bidder with highest value in value profile $u$ -- to bid at least $u_{[2]} = v_{[1]} - \eps$. Such a bid will help the auctioneer leverage almost full surplus from $v$.  

For convenience, we will call $u$ the \emph{auxiliary vector} of $v$.  We now show that for any realized value profile $v \in \V$, such an auxiliary vector $u$ can be constructed. We divide the argument into four cases. Recall that we use $\rho^* = \max_{\theta,i} v_i(\theta)$ to denote the maximum value among all bidders' values at all states, and $i^*$ denotes the bidder who has  value $\rho^*$. Similarly, $\rho^{**} = \max_{i \not = i^*; \, \theta} v_i(\theta)$ is the maximum among all possible valuations, excluding bidder $i^*$'th, and $i^{**} \not = i^*$ denote a bidder who has the value $\rho^{**}$.

\begin{itemize}
	\item {\bf Case 0}: the bidder of  the largest value is not unique.  Let $I$ denote the set of bidders with the largest value $v_{[1]}$ in profile $v$. So $|I| > 1$.  Construct $w^1, w^2 \in \RR^n$ as follows. Let $w^1 =v$; Let $\tilde{i}$ be the largest index in $I$ and define $w^2_{i}=v_i$ for all $i \not \in I$, $w^2_{\tilde{i}} = v_{\tilde{i}}$ and $w^2_{i} = 0$ for all $i \in I, i \not = \tilde{i}$. Note that $w^1, w^2 \in \V = \prod_{i=1}^n \V_i$, thus have strictly positive probabilities by our assumption that $\lambda$ supports on the entire set $\V$.   It is easy to verify that a mixture of $w^1$ with probability $ 1 -\frac{\epsilon}{v_{[1]}}$ and $w^2$ with probability $\frac{\epsilon}{v_{[1]}}$ results in an expected value profile $u$ such that $u_i = v_i$ for all  $i \not \in  I$ and  $u_{\tilde{i}} = v_{[1]}$, $u_i = v_{[1]} - \epsilon$ for all $i \in I, i \not = \tilde{i}$. Such a $u$ satisfies the conditions in Lemma \ref{lem:private} for any arbitrarily small $\eps$. For the following analysis, we assume the vector $v$ has a unique maximum value. 
	
	\item {\bf Case 1}: $v_{i^*} \leq \rho^{**}$ and  $v_{i^*}$ is \emph{not} the largest value in $v$. Let $\tilde{i}$ be the bidder with the largest value $v_{[1]}$, thus $v_{\tilde{i}} = v_{[1]}>v_{i^*}$. Construct $w^1, w^2 \in \RR^n$ as follows. Let $w^1_i=v_i$ for all $i \not =  i^*,\tilde{i}$, and $w^1_{\tilde{i}} = v_{[1]}$, $w^1_{i^*} = \rho^*$; let $w^2_i=v_i$ for all $i \not = \tilde{i},i^*$, and $w^2_{\tilde{i}} = 0$, $w^2_{i^*} = \rho^*$.  It is easy to see that a mixture of $w^1$ with probability $ 1 -\frac{\epsilon}{v_{[1]}}$ and $w^2$ with probability $\frac{\epsilon}{v_{[1]}}$ results in an expected value profile $u$ such that $u_i = v_i$ for all  $i \not = \tilde{i},i^*$ and  $u_{\tilde{i}} = v_{[1]} - \epsilon$, $u_{i^*}= \rho^*$. 

	\item {\bf Case 2}: $v_{i^*} \leq \rho^{**}$ and  $v_{i^*}$ is  the largest value in $v$, thus $v_{i^*} = v_{[1]}$. Construct $w^1, w^2 \in \RR^n$ as follows. Let $w^1_i=v_i$ for all $i \not = i^{**}, i^*$, and $w^1_{i^{**}} = \rho^{**}$, $w^1_{i^*} = \rho^*$; let $w^2_i =v_i$ for all $i \not = i^{**},i^*$, and $w^2_{i^{**}} = \rho^{**}$, $w^2_{i^*} = 0$. A mixture of $w^1$ with probability $ \frac{ v_{[1]}-\epsilon}{\rho^*}$ and $w^2$ with probability $ \frac{ \rho^*-v_{[1]}+\epsilon}{\rho^*} $ results in an expected auxiliary vector $u$ such that $u_i = v_i$ for all  $i \not = i^{**},i^*$ and  $u_{i^{**}} = \rho^{**}$, $u_{i^*} = v_{[1]} - \epsilon$.  

	\item {\bf Case 3}: $v_{i^*} > \rho^{**}$, which happens with probability $\sum_{\theta: v_{i^*}(\theta)>\rho^{**}} \lambda_{\theta}$. This implies that bidder $i^*$ must have the highest value in $v$. Construct $w^1, w^2 \in \RR^n$ as follows.  Let $w^1_i=v_i$ for all $i \not = i^{**}, i^*$, and $w^1_{i^{**}} = \rho^{**}$, $w^1_{i^*} = \rho^*$; let $w^2_i=v_i$ for all $i \not = i^{**},i^*$, and $w^2_{i^{**}} = \rho^{**}$, $w^2_{i^*} = 0$. A mixture of $w^1$ with probability  $ \frac{ \rho^{**}-\epsilon}{\rho^{*} }$ and $w^2$  with probability $\frac{ \rho^* - \rho^{**} + \epsilon}{\rho^{*} }$ results in an expected auxiliary vector $u$ such that $u_i = v_i$ for all  $i \not = i^{**},i^*$ and  $u_{i^{**}} = \rho^{**}$, $u_{i^*} = \rho^{**} - \epsilon$.  Therefore, the second highest value in $u$ is $\rho^{**} - \epsilon$, which is the  revenue at the worst equilibrium (Lemma \ref{lem:private}). Since the full surplus  is $v_{i^*}$,  the revenue is at most $v_{i^*} - \rho^{**} + \epsilon$ away from  the full surplus. The aggregated surplus loss from this case is thus $\sum_{\theta:v_{i^*}(\theta)>\rho^{**}} \lambda_{\theta}[v_{i^*}(\theta) - \rho^{**} +\eps]$, which is precisely the excluded term in Theorem \ref{thm:private}, up to an $\eps$. 
\end{itemize}

To sum up,  the auxiliary vector $u$ can be constructed for any value profile $v$. Moreover, the construction corresponds to a signaling scheme that mixes two properly chosen value profiles $w^1, w^2$, as in the above case analysis. We can then use the private signaling scheme in Lemma \ref{lem:private} -- i.e., revealing no information to the bidder of the highest value in  the auxiliary vector $u$ and full information to all other bidders -- to extract almost full surplus from $v$. Note that the excluded term in Theorem \ref{thm:private} is mainly due to Case $3$.  Notice that the probability mass of the auxiliary vector $u$ can be arbitrarily small so long as it is positive, therefore the fact that a certain value profile is used for constructing auxiliary  vectors has negligible effects when we consider the revenue extracted from that value profile. 

The above argument only shows the existence of a private signaling scheme that extracts almost the full surplus. Next, we prove that this  scheme can be implemented efficiently. In particular, given any realized value profile $v(\theta)$, we will need to efficiently decide which signal to send with what probability. The main issue here is, when $v(\theta)$ is needed for constructing auxiliary vectors, we have to identify all the value profiles that use $v(\theta)$ to construct their auxiliary vectors, and signal accordingly. Our next lemma shows that this can be done efficiently, thus completing the proof of Theorem \ref{thm:private}.

\begin{lemma}
	The private signaling scheme above can be implemented in $\poly(n,\sum_i |\V_i|)$  time. 
\end{lemma}
\begin{proof}
	We will show that given any realized value profile, we can explicitly check whether it is used for constructing auxiliary vectors, and if so, of which value vectors.  For convenience, we call the value profiles that are used for constructing auxiliary vectors the \emph{basic vectors}.  A key observation is that for any given  $v(\theta)$, we can explicitly identify all value profiles (if any) that use  $v(\theta)$ to construct their auxiliary vectors, and there are only  $\poly(\sum_i |\V_i|)$ of them.  This follows a tedious discussion, which essentially reverse the above case analysis. We omit the details here.

	With the above observation, we can now easily construct the desired private signaling scheme as follows.  Given any realized profile $v(\theta)$, let $u$ be its auxiliary value profile as we constructed and $i$ be the bidder of the highest value in $u$. If $v(\theta)$ is not a basic vector, we simply reveal $v(\theta)$ to all the bidders except that we tell bidder $i$ the value profile is $v(\theta)$ with probability $1-\delta$ and $u$ with probability $\delta$, for some preset small enough $\delta$. If $v(\theta)$ is a basic vector. We first find all the ($\poly(\sum_i |\V_i|)$ many) value profiles that use $v(\theta)$ to construct auxiliary vectors.  We can compute the probability needed from $v(\theta)$ for each such value profile, and send the signal with corresponding probabilities by referring to the case analysis above. For the left probability mass, we simply treat $v(\theta)$ as a basic vector and signal accordingly.
\end{proof}

\end{document}